\documentclass{article}

\usepackage[T1]{fontenc}
\usepackage{graphicx}
\usepackage{hyperref}
\usepackage[margin=1.25in]{geometry}
\usepackage{color}

\usepackage{fancyhdr}
\usepackage{subcaption}
\usepackage{amssymb, amsmath}
\usepackage{pdfcomment}
\usepackage{longtable}

\usepackage{amsthm,amssymb,amsmath}
\usepackage{algpseudocode}
\usepackage{algorithm}

\theoremstyle{plain}
\newtheorem{theorem}{Theorem}[section]
\newtheorem{proposition}[theorem]{Proposition}
\newtheorem{observation}[theorem]{Observation}
\newtheorem{open}[theorem]{Open Question}

\newtheorem{lemma}[theorem]{Lemma}
\newtheorem{corollary}[theorem]{Corollary}

\theoremstyle{definition}
\newtheorem{definition}[theorem]{Definition}

\def\full{false}

\begin{document}

\title{Reconstructing Sets of Strings from Their $k$-way Projections: Algorithms \& Complexity}
\author{Elise Tate\footnote{Department of Computer Science, University of Colorado Boulder, USA, \texttt{jordan.tate@colorado.edu}}\; and
Joshua A. Grochow\footnote{Departments of Computer Science and Mathematics, University of Colorado Boulder, USA, \texttt{jgrochow@colorado.edu}, ORCID: \href{https://orcid.org/0000-0002-6466-0476}{0000-0002-6466-0476}}
}

\maketitle              

\begin{abstract}
Graphs are a powerful tool for analyzing large data sets, but many real-world phenomena involve interactions that go beyond the simple pairwise relationships captured by a graph. In this paper we introduce and study a simple combinatorial model to capture higher order dependencies from an algorithms and computational complexity perspective. Specifically, we introduce the String Set Reconstruction problem, which asks when a set of strings can be reconstructed from seeing only the $k$-way projections of strings in the set. This problem is distinguished from genetic reconstruction problems in that we allow projections from any $k$ indices and we maintain knowledge of those indices, but not which $k$-mer came from which string. We give several results on the complexity of this problem, including hardness results, inapproximability, and parametrized complexity. 

Our main result is the introduction of a new algorithm for this problem using a modified version of overlap graphs from genetic reconstruction algorithms. A key difference we must overcome is that in our setting the $k$-mers need not be contiguous, unlike the setting of genetic reconstruction. We exhibit our algorithm's efficiency in a variety of experiments, and give high-level explanations for how its complexity is observed to scale with various parameters. We back up these explanation with analytic approximations. We also consider the related problems of: whether a single string can be reconstructed from the $k$-way projections of a given set of strings, and finding the largest $k$ at which we get no information about the original data set from its $k$-way projections (i.e., the largest $k$ for which it is ``$k$-wise independent''). 
\end{abstract}

\section{Introduction}
Your friend has a complex data set, but you can only view a subset at a time. Can you figure out what their data set is? There are many different precise instantiations of this general type of question. For example, the ability to extrapolate from an incomplete data set gets at the core of many tasks in machine learning and (artificial) intelligence. 

As another example, this question can also be used to understand when---and which---higher-order interactions are most relevant to a given phenomenon about a complex, multi-part system. That is, even if the entire data set can be viewed or accessed globally, by artificially restricting ones attention to subsets of size $k$, one can see what information about the data set can be gleaned from such $k$-subsets beyond that which can be gleaned from the $(k-1)$-subsets, revealing the ``$k$-th order interactions'' in the data.

In this paper we consider a novel problem that is a combinatorial interpretation of this question, which we refer to as the \textsc{String Set Reconstruction} problem. Consider an $m \times n$ matrix. For each fixed set of $k$ out of the $n$ columns, we get a set of length-$k$ subsequences, called a $k$-way projection. Can all the $k$-way projections be used to reassemble the original data set? How does this change as we change the value of $k$? In this framework, we know which columns we are viewing, but we only see the \emph{set} of those length-$k$ strings. In particular, we do not know which length-$k$ string came from which original length-$n$ string. 
 
This problem is natural from the perspective of studying systems by their higher-order interactions (beyond pairwise) (see, e.\,g., the surveys \cite{GrochowSFIReview,TorresEtAl,BattistonEtAl,BickEtAl}). In particular, the \textsc{String Set Reconstruction} problem can be viewed as a particular instance of asking when $k$-th order interactions suffice to completely determine a given data set. This may be especially beneficial in the context of data sets that one may collect in real-world experiments, where collecting a data set dense enough for traditional machine learning algorithms is difficult or impossible. One example that was inspirational in our approach to this problem is a data set of conflicts among a group of macaques \cite{DanielsKrakauerFlack,LeeDanielsFlackKrakauer}. This data set was collected through over one hundred hours of field work recording every fight that occurred among a group of 48 adult macaques, resulting in 1087 recorded fights. Each fight is recorded as a binary string of length 48 where each individual macaque corresponds to an index value which is a 1 if that macaque was in the fight, and a 0 if it was not. Even though this data set took considerable effort to collect, it is clear that 1087 strings out of the $2^{48}$ possible strings is a very small portion, and to collect enough data to change that would be impossible---indeed, even if they could be observed for an arbitrary long period of time, it's not clear that anywhere close to all $2^{48}$ possible fights would arise. A second example we have considered include disease spread on a floor of a dorm, where each student is either a 0/1 depending on if they are healthy or ill on a given week, and each string corresponds to one week in the semester. Yet another example is a psychological survey on several hundred participants and their perception of discrimination, where each participant corresponds to one string and each digit of the string is their answer for a particular question. All of these examples are data sets with relatively few data points (compared to the maximum possible given the dimension of the data set) but where higher-order interactions come into play.\footnote{We also note that in the regime where the number of strings is exponential in their length, then a runtime that is linear in the input size is the same as being exponential in the string length. With that generous a runtime, much simpler algorithms would suffice in theory, but those algorithms would tend to not be useful in practice.}

In this paper, our main results are two-fold. In terms of complexity, we give a nearly-complete complexity classification of the problems we study, from the viewpoints of worst-case complexity, (in)approximability, and parametrized complexity. In terms of algorithms, we introduce a new algorithmic technique for these problems, generalizing the notion of ``overlap graphs'' from genetic reconstruction \cite{rizzi19}. We demonstrate the efficiency of our algorithm in practice through a series of experiments, as well as heuristic mathematical analysis of how the runtime is observed to scale with various parameters.

A string $x$ is contained in our reconstruction from the $k$-way projections ($k$-reconstruction for short) iff all of its size-$k$ subsequences (not necessarily contiguous) are contained in the corresponding $k$-way projection of our data set. That is, given a set of strings $S$ and another string $x$, $x$ is in the $k$-reconstruction of $S$ iff for each $I \subseteq [n]$, there is a $y \in S$ such that $y|_{I} = x|_{I}$, where $y|_{I}$ denotes the characters in $y$ at the positions indexed by the set $I$. 
As such, for any value of $k$ we will recover all of the strings in the original data set. 
It is possible, however, that for smaller values of $k$ we may recover additional strings as well: there may be strings that do not contradict any of our subsequences in their $k$-way projections, and as such they will be included in our reconstruction. We can either think of $k$ as a parameter that we increase from $1$ upward to eliminate more and more potential strings, or one that we decrease from $n$ downward until we start being unable to eliminate strings outside of the input set. 

This same concept shows up under the name ``$k$-limit'' or ``local limit'' in the proof of certain circuit complexity lower bounds \cite{HJP,GRSS}, but we are not aware of prior studies of the algorithmic problems we study. 

Within this analytical framework, perfect reconstruction is just one extreme. In the next several subsections we discuss several related problems that are interesting both for their potential applications in data analysis, and for the algorithmic questions they raise. There are also several variations to these problems that are significant both in terms of the computational complexity of the problem, as well as the types of algorithms that can be used to tackle the problem. 

\subsection{The problems we study}
In this paper we assume that a set $S$ of strings of length $n$ is given as part of the input.\footnote{Although it would also be interesting to study the variant where only the $k$-way projections are given, and not the set $S$, the complexity can be quite different, and for consistency in this paper we focus on the case where $S$ is given, leaving the other setting for the full version of the paper and future work.}
For simplicity, throughout this paper we focus on the alphabet $\{0,1\}$, though the exact characters are irrelevant. Almost everything in the paper immediately generalizes to alphabets of arbitrary fixed size; when this is not the case we will note it.

\paragraph{Point of perfect reconstruction.}

The most apparent question, and the one that we will primarily focus on throughout this paper, is determining the value of $k$ that allows us to perfectly reconstruct our original data set. We give this value a name:

\newcommand{\Recon}{Recon}
\begin{definition}[with Yoav Kallus] \label{def:recon}
Given an alphabet $A$ and $S \subseteq A^n$, $x \in A^n$ is \emph{$k$-reconstructible} from $S$ if for every $I \subseteq [n]$ with $|I|=k$, $\exists s \in S$ such that $s|_I = x|_I$. The \emph{$k$-reconstruction} of $S$ is the set of all strings $k$-reconstructible from $S$, denoted $\Recon_k(S)$. A set $S$ is called \emph{$k$-reconstructible} if $S = \Recon_k(S)$. The \emph{point of perfect reconstruction} is the least integer $k$ such that $S=\Recon_k(S)$.
\end{definition}

We will refer to the computational problem of computing the point of perfect reconstruction given $S$ as the \textsc{Perfect Reconstruction} problem. This value is an interesting structural parameter of the set of strings $S$, offering some insight into higher-order interactions in $S$. See Fig.~\ref{fig:perfRecon} for an example.

\begin{figure}
    \centering
    \includegraphics[width=\textwidth]{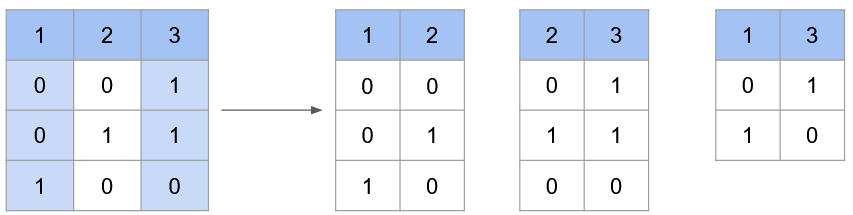}
    \caption{\textbf{An example of a set $S$ of three length-$3$ strings whose point of perfection reconstruction is $< 3$.} To the left is the starting data set $S = \{001, 011,100\}$, where the strings correspond to rows after the first, and the indices 1,2,3 are listed in the first row. To the right are the 2-way projections for index sets $\{1,2\}$, $\{2,3\}$, and $\{1,3\}$, respectively. From the $\{1,3\}$ projection, we see that indices 1 and 3 completely define one another: $x_1 = 1-x_3$.
    This, coupled with the $\{1,2\}$ projection, allows us to fill in the 2nd column and perfectly reconstruct the original data set.    
    }
    \label{fig:perfRecon}
\end{figure}

There are sets whose point of perfect reconstruction is $n$, the largest possible value. For example, $S = \{10^{n-1}, 010^{n-2}, 0010^{n-3}, \dotsc, 0^{n-1} 1\}$ (the ``standard basis vectors'' of an $n$-dimensional vector space). Since $0^{n-1}$ appears in all  $(n-1)$-way projections, we have that $0^n \in \Recon_{n-1}(S) \backslash S$.

Computing the point of perfect reconstruction can be done by a brute-force, exponential-time algorithm. Namely, the algorithm takes each possible string of length $n$ and generates all $n \choose k$ of its subsequences, and for each check those subsequences against the list of subsequences from the original data set (in the same $k$ positions) to see if it is in the input set or not. The runtime of this procedure is $O(2^n m \binom{n}{k})$, where $m$ is the number of starting strings.

\paragraph{Point of no information ($k$-wise independence).}

To better understand the dependence on $k$, we can look at the ``opposite extreme'' of the point of perfect reconstruction, namely, the largest $k$ for which the $k$-way projections give essentially \emph{no information} about the original set of strings. To clarify what we mean by ``no information'', we begin with a straightforward observation:

\begin{observation} \label{obs:1recon}
Let $A$ be a finite alphabet, $S \subseteq A^n$, and for $i \in [n]$ let $A_i = S|_{\{i\}}$ be the 1-projections. Then $S$ is 1-reconstructible iff $S = A_1 \times A_2 \times \dotsb \times A_n$, iff $|S| = \prod_{i=1}^n |A_i|$.
\end{observation}

Here the Cartesian product $A_1 \times \dotsb \times A_n$ denotes the set of all length-$n$ strings $x \in A^n$ such that $x_i \in A_i$ for all $i$. 

In particular, every character occurs in every position if and only if the 1-reconstruction of $S$ is equal to all of $A^n$. This leads us to the following definition:

\begin{definition}
Given a set of strings $S \subseteq A^n$ over an alphabet $A$, the \emph{point of no information} is the largest $k$ such that the $k$-reconstruction of $S$ is $A^n$.
\end{definition}

A probability distribution on $\{0,1\}^n$ that is $k$-wise independent is supported on a set $S$ of strings whose point of no information is $\geq k$. The converse need not hold, however. A tighter connection with $k$-wise independence can be obtained if one requires not only that the $k$-way projection sets contain all $2^k$ possibilities, but that each such possibility comes from the same number of strings in the original set; equivalently, looking at the $k$-way projection \emph{multi}-sets, rather than just sets. We leave exploring this connection further for future work.

There will always be no information at $k=1$ unless there is some column that doesn't use all the characters; the point of no information may also be much larger, as in the case of $S$ be the set of length-$n$ strings of even parity, where the point of no information is $k=n-1$. 

As with the point of perfect reconstruction, there is a simple brute-force algorithm for this problem. The point of no information is at least $k$ iff every length-$k$ string occurs in every $k$-way projection. 
Thus we can simply check each of the $n \choose k$ lists of subsequences to make sure each set contains all possible $2^k$ subsequences. The runtime is $O\left(\binom{n}{k} |S|\right)$. 

\paragraph{String containment.}
When we are trying to determine the difference between perfect reconstruction and nearly perfect reconstruction, we come across a critical question: given some string $x$, how can we determine if $x$ is in $\Recon_k(S)$? We will refer to this as \textsc{String Containment} problem. More precisely, given some string $x$ and a data set $S$, can we determine some $k$-length subsequence that is present in $x$ but not in any string in $S$? Such a subsequence exists iff $x \notin \Recon_k(S)$. The \textsc{String Non-Containment} problem turns out to be isomorphic to the classic \textsc{Hitting Set} problem (Thm.~\ref{SCtoHS}), in a way that preserves approximability and multi-parameter complexity.

\subsection{Motivating our algorithmic technique with a larger example} \label{sec:overlap_intro}
We introduce our new algorithm with a larger example that exhibits several of the phenomena we've seen so far. 
Fig.~\ref{fig:largeOLG} shows a data set $S$ where $k=2$ contains no information, $k=3$ has some information, and $k=4$ perfectly reconstructs $S$: $\{0,1\}^5 = \Recon_1(S) = \Recon_2(S) \supsetneq \Recon_3(S) \supsetneq \Recon_4(S) = S$.

$k=2$: One can verify by ``brute force inspection'' that all $\binom{5}{2}$ 2-way projections contain all 4 possible subsequences of length 2.

$k=3$: As there are only 7 strings, $k=3$ must contain some constraints or information, as there aren't enough strings to exhibit all 8 possible subsequences of size 3. We illustrate the connection with \textsc{Hitting Set} from \S\ref{sec:complexity} here, by showing that $x=11110 \in \Recon_3(S) \backslash S$, so $k=3$ is not the point of perfect reconstruction. $x$ differs from $01110 \in S$ only in position 1, so any index set $I$ such that $x|_I \notin S|_I$ must have $1 \in I$. Similarly, $x$ differs from $11111 \in S$ in index $5$, so we have $\{1,5\}$ in any distinguishing index set $I$. The three possibilities are $I=125$, $135$, or $145$, and the subsequence of $x$ for each of these is $110$. For $I=125$ or $135$, $11100 \in S$ has restriction $110$, agreeing with $x$; and similarly for $I=145$, $10010 \in S$ agrees with $x$. Thus $x \in \Recon_3(S)$.

$k=4$: To show $\Recon_4(S)=S$, we give an example of our main algorithmic technique: the overlap graph. See Fig.~\ref{fig:largeOLG}. Using the ordering $1,2,3,4,5$ on the columns, the overlap graph will have levels, one for each cyclically contiguous set of $k$ indices in our ordering: $[1,2,3,4]$, $[2,3,4,5]$, $[3,4,5,1]$, and so on. The nodes in each level correspond to the projections of the strings in $S$ onto each size-$k$ index set. There is a directed edge between nodes at one level and the next if they agree on the overlapping $k-1$ indices. From this setup, the length-$n$ directed cycles in this graph correspond bijectively to the set of strings that would be reconstructed if only looking at these particular $k$-way projections (the overlap graph only sees $n$ out of the $\binom{n}{k}$ possible $k$-way projections). The overlap graph may have extra cycles that can be eliminated by other $k$-subsets of indices not used in the graph---which can be done using calls to \textsc{Hitting Set}---but the overlap graph does not eliminate any $k$-reconstructible strings. In Fig.~\ref{fig:largeOLG}, since the number of directed cycles (yellow highlights) is the same as the number of strings we started with, at $k=4$ all strings outside $S$ are excluded by using the $k$-way projections in this overlap graph, so $\Recon_4(S)=S$. We give more algorithmic details of the overlap graph in \S\ref{sec:algorithm}.

\begin{figure}
    \centering
    \includegraphics[width=0.25\textwidth]{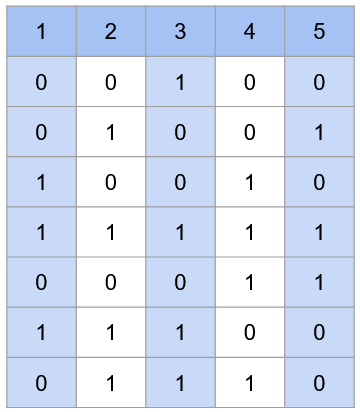}
    \includegraphics[width = 0.72\textwidth]{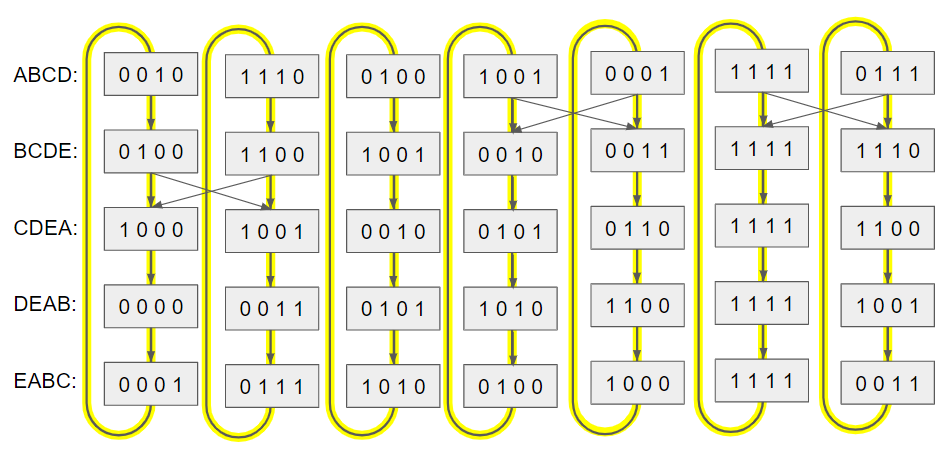}
    \caption{
    An example containing 7 strings of length 5 with the indices enumerated across the top of the table, along with one of its overlap graphs for $k=4$. 
    The cycles of length 5 are highlighted yellow; the number of $5$-cycles is the same as $|S|$, so $k=4$ is the point of perfect reconstruction.  The arrangement of vertices into columns is not inherent, but merely for convenience.}
    \label{fig:largeOLG}
\end{figure}

\subsection{Related work} \label{sec:related}

Overlap graphs are commonly used in DNA sequence assembly methods (see, e.\,g., \cite{rizzi19}). The core concept is to make a graph where each node represents one sequence read, and edges are drawn between nodes with matching prefix/suffix pairs. When using this to reconstruct a DNA sequence, you are generally constructing a single string from pieces where you do not have the location information for those pieces. In contrast, we have the location information but do not know how many sequences we are producing. 

We must address these critical differences in our adaptation of the overlap graph. First, we are reassembling many strings (of the same length) rather than only one of unknown length; second, we maintain knowledge of our indices; third, we have access to all subsequences instead of only contiguous ones. 
These first two features allow us to develop a more rigid structure in our overlap graphs when compared to their bioinformatics counterparts, while the third feature will come into play in our preprocessing steps (\S\ref{sec:preprocessing}).

The benefit of knowing the indices of our subsequences is observed in the layered structure of our overlap graph. In overlap graphs seen in bioinformatics, there is no layered structure and edges are drawn between two nodes with overlaps of all sizes. By knowing our indices we can instead draw edges only between adjacent layers of the overlap graph without losing information. 
We also benefit from complete saturation of our overlap graph; we can observe every window of indices rather than risking there being a gap, as there might be in DNA sequence reassembly. This means we can be confident that our overlap graph contains all of our information.

We can also consider the problem of computing $\Recon_k(S)$ from $S$. The set $\Recon_k(S)$ consists precisely of the solutions to a corresponding $k$-CSP instance, in which the underlying constraint hypergraph is the complete $k$-uniform hypergraph, and the constraints are the $k$-way projections. This connects to the literature on counting and enumerating CSPs (e.\,g., \cite{Petke,Zivny,CKS,Feder}). However, we note that in the context of CSPs, typically the set of allowed constraints is fixed in advance, and which sets of $k$ variables they are applied to varies with the instance. In contrast, in our setting, the constraints are always applied to all $k$-tuples of the $n$ variables, but what constraints we use depend on the input. 

\section{Computational complexity analysis} \label{sec:complexity}

\subsection{Preliminaries for complexity results} \label{sec:prelim}
We assume the reader is familiar with standard complexity classes $\mathsf{NP}, \mathsf{coNP}$ and the polynomial hierarchy (e.\,g., \cite{sipser,CKS}), as well as the multi-parameter complexity classes $\mathsf{FPT}$ and $\mathsf{W[2]}$ (e.g., \cite{Downey}). The complexity class $\mathsf{DP}$ \cite{Papadimitriou84} consists of those languages $L$ such that there are two languages $L_1, L_2 \in \mathsf{NP}$ with $L = L_1 \backslash L_2 = \{x : x \in L_1 \text{ and } x \notin L_2\}$. For the Strong Exponential Time Hypothesis (SETH) see \cite{Calabro09,IP01}. For the Unique Games Conjecture see \cite{Khot02}.
Two sets $A, B \subseteq \Sigma^*$ are \emph{p-isomorphic} if there is a bijection $f \colon \Sigma^* \to \Sigma^*$ such that both $f$ and $f^{-1}$ are computable in polynomial time, and $f(A)=B$. In this case we write $A \cong_p B$. 

For the proof of Theorem~\ref{SCtoHS} we will need the following additional preliminaries. A language $A$ is \emph{p-paddable} if there are polynomial-time computable functions $S,D$ such that, for all strings $x,y$, $S(x,y) \in A \Leftrightarrow x \in A$ and $D(S(x,y)) = y$. The idea is that $S$ adds ``spare'' parts to an instance of $A$ to make it longer (``pad it'') without changing its membership in $A$, and $D$ shows that the ``spare'' parts are uniquely and efficiently identifiable. In particular, the existence of $D$ implies that for each $x$, the map $y \mapsto S(x,y)$ is 1-to-1. We will use the following now-classic theorem:

\begin{theorem}[{\cite[Thm.~1]{Berman75}}] \label{thm:BH}
If $A$ is $\mathsf{NP}$-complete \& p-paddable, then $A \cong_p SAT$. 
\end{theorem}

\subsection{String (Non-)Containment is Hitting Set}
We will show that \textsc{String Non-Containment} is essentially the same problem as \textsc{Hitting Set}, from the viewpoints of worst-case, approximation, fixed-parameter, and fine-grained complexity. The \textsc{Hitting Set} decision problem is: given a collection $S$ of subsets of a finite universe $U$, and an integer $k$, does there exist a set $H \subseteq U$ of size $k$ such that every set in $S$ contains at least one element of $H$? 

To state some of our hardness results, we use an additional parameter, beyond $k$. Namely, given a set $S \subseteq \{0,1\}^n$ of strings, let
\[
d(S) := \max_{x \in \{0,1\}^n} \min_{s \in S} d_{Hamming}(x,s) = \max_{x} d_{Hamming}(x,S),
\]
where $d_{Hamming}(x,s)$ denotes the Hamming distance between $x$ and $s$.

\begin{theorem}
\label{SCtoHS}
\textsc{String Non-Containment} is essentially the same problem as \textsc{Hitting Set}. In particular: 
\begin{enumerate}
    \item The two are p-isomorphic, and thus  \textsc{String Containment} is $\mathsf{coNP}$-complete.
    \item The minimum $k$ such that $x \notin \Recon_k(S)$ can be $d(S)$-approximated in polynomial time. Assuming the Unique Games Conjecture, the factor of $d(S)$ is optimal.     
    \item Finding the exact minimum $k$ for which $x \notin \Recon_k(S)$ is $\mathsf{DP}$-complete. 

    \item \textsc{String Containment}, when parametrized by $k$ and $d(S)$, is in $\mathsf{FPT}$.
    
    \item When parametrized only by $k$, \textsc{String Containment} is $\mathsf{coW[2]}$-complete.    

    \item SETH is equivalent to: For all $\varepsilon > 0 $, there is a $d$ such that \textsc{String Containment} with $d(S) \leq d$ is not solvable in time $O(2^{(1-\varepsilon)n})$. 
\end{enumerate}
\end{theorem}

\begin{proof}
We start by showing the equivalence of the two problems, and then use this equivalence to show the various parts of the theorem. 

Given an instance $(S,x,k)$ of \textsc{String Non-Containment}, define the following instance of \textsc{Hitting Set}: 
\begin{itemize}
\item If $S$ is a set of length-$n$ strings, then our universe $U = [n]$.

\item For each $s \in S$, add a set $T_s$ which is the subset of $[n]$ consisting precisely of the indices $i$ where $x_i \neq s_i$.

\item The $k$ for \textsc{Hitting Set} is the same as the $k$ for \textsc{String Non-Containment}.
\end{itemize}
We show that there is a hitting set of size $k$ for the above instance if and only if $x$ is not contained in the $k$-reconstruction of $S$. Suppose there is a hitting set $H \subseteq [n]$ of size $k$. We claim that $H$ consists of the indices of a $k$-window in which $x|_H$ does not appear. Since $H$ is a hitting set, for each $s \in S$ there is some $i$ such that $i \in T_s \cap H$. By the definition of $T_s$, this means that $x_i \neq s_i$. Thus, in the window indexed by $H$, $x$ differs from every $s \in S$, so $x$ is not in $\Recon_k(S)$.

Conversely, given a \textsc{Hitting Set} instance $(T_1, \dotsc, T_m; k)$ where $T_i \subseteq U=[n]$ for all $i$, we build the following instance of \textsc{String Non-Containment}. 
\begin{itemize}
\item The strings in our set will be the $\{0,1\}$-indicator functions of the sets $T_i$.

\item The string $x$ will be the all-$0$s string.

\item The $k$ is the same $k$.
\end{itemize}
Suppose the string $x$ is not contained in the $k$-reconstruction of $S$. Then there is a $k$-window, $H \subseteq [n]$, such that $x=0^n$ differs from every $s_j \in S$ at some index $i_j \in H$. But since $x=0^n$, differing from $x$ is the same as saying that $s_i = 1$. In other words, if we consider $H$ as a subset of $[n]$, then $i \in H \cap T_j$, and thus $H$ is a hitting set for the original hitting set instance.

Now we turn to the statements in the theorem:
\begin{enumerate}
\item It is standard that \textsc{Hitting Set} is p-isomorphic to \textsc{CNF-SAT}, so it suffices for us to show the same for \textsc{String Non-Containment}. We will exhibit functions $S,D$ so that Thm.~\ref{thm:BH} applies (see Sec.~\ref{sec:prelim} for the properties needed of $S$ and $D$). For each string $x$, let $d(x)$ denote the string where each character of $x$ is doubled: $x_0 x_0 x_1 x_1 x_2 x_2 \dotsb x_n x_n$. Identify the alphabet $\Sigma$ with the integers modulo $m = |\Sigma|$, and let $e(x)$ be the string where each character is replaced by a pair consisting of the original character $x_i$ followed by $x_i+1 \pmod{m}$. For example, when $m=3$, we have $d(021) = 00\ 22\ 11$ and $e(021)=01\ 20\ 12$ (spaces added for easier visual parsing). Define $S(S,x,k; y) := (S'(y), d(x)e(y), 2k)$, where $d(x)e(y)$ denotes the concatenation, and $S'(y)$ is as follows: for each string $s \in S$, the string $d(s)e(y)$. That $S$ is computable in polynomial-time is clear. The reason we doubled digits and encoded $y$ as we did was to construct the inverse/un-padding function $D$: $D(S'(y),d(x)e(y),2k)$ looks at its second string, and because of the encoding $d(x)e(y)$, can tell uniquely what $y$ was: look at the longest suffix consisting of pairs of digits of the form $(a,(a+1 \pmod{m}))$, and take the first digit of each pair. 

It remains to show that $x$ is in the $k$-reconstruction of $S$ if and only if $d(x)e(y)$ is in the $2k$-reconstruction of $S'(y)$. Let $n = |x|$ and $m = |y|$. Say that a $2k$-window is ``padding-avoiding'' if all its indices are at most $2n$. Since the same string $e(y)$ is appended both to $d(x)$ and to $d(s)$ for all $s \in S$, it follows that if $W \subseteq [2n+2m]$ is any window, then $d(x)e(y)$ does not agree with $S'(y)|_W$ if and only if $d(x)$ does not agree with $S'(y)|_{W \cap [2n]}$. Thus, if $d(x)e(y)$ disagrees with $S'(y)$ in some window of size $\leq 2k$, then it disagrees with $S'(y)$ in some padding-avoiding window of size at most $2k$ (and vice versa, but the opposite direction is immediate). We thus restrict our attention to padding-avoiding windows.

Next, since the columns in $d(x)$ are simply doubled versions of those of $x$, we get a correspondence between padding-avoiding windows of $S'(y)$ and arbitrary windows of $S$. Namely, if $I' \subseteq [2n]$ is a padding-avoiding window, then $I := \{\lfloor i/2 \rfloor : i \in I'\} \subseteq [n]$ is a window such that $x|_I \notin S|_I \Leftrightarrow d(x)e(y)|_{I'} \notin S'(y)|_{I'}$. Furthermore, $|I'| \leq 2|I|$. Thus, we get that $x$ is in the $k$-reconstruction of $S$ iff $d(x)e(y)$ is in the $2k$-reconstruction of $S'(y)$, as claimed.

\item From the proof of the above reductions, it is clear that the minimum $k$ in which $x$ is not in the $k$-reconstruction of $S$ is the same as the smallest hitting set in the corresponding \textsc{Hitting Set} instance. Thus, this same value can be approximated to the same extent that the size of the minimum hitting set can be approximated. The result is then immediate from the main result of Khot \& Regev \cite[\S~4]{KhotRegev}. Note that, in our reduction, the Hamming distance $d(x,S) \leq d(S)$ corresponds precisely to the size of the sets in the \textsc{Hitting Set} instance. To match up their language with ours: Khot \& Regev use the language of vertex cover in a hypergraph, where the sets in the \textsc{Hitting Set} instance give the hyperedges in the hypergraph, so a $k$-uniform hypergraph is the same a \textsc{Hitting Set} instance in which all the sets have size $k$.

\item The problem under consideration here is: given $(S,x,k)$, is $k$ the minimum value such that $x$ is not in the $k$-reconstruction of $S$? Since the \textsc{String Non-Containment} decision problem is in $\mathsf{NP}$ (the witness exhibits the window where $x$ differs from all strings in $S$), we see that the exact minimum problem is in $\mathsf{DP}$: $k$ is the exact minimum iff $x \in \Recon_k(S)$ and $x \notin \Recon_{k-1}(S)$.

To see $\mathsf{DP}$-hardness, we reduce from a known $\mathsf{DP}$-complete problem: \textsc{Exact Max Independent Set}. Here the input is $(G,k)$ where $G$ is a graph, and the output is to decide whether $k$ is the size of the largest independent set in $G$. As is well-known, a set $I \subseteq V(G)$ is an independent set if and only if $V(G) \backslash I$ is a vertex cover, so $k$ is the exact max independent set size iff $n-k$ is the exact min vertex cover size. Thus \textsc{Exact Min Vertex Cover} is also $\mathsf{DP}$-complete. Finally, as \textsc{Vertex Cover} is the special case of \textsc{Hitting Set} where all the sets have size 2, from the reductions above we then immediately get that \textsc{Exact Min String Non-containment} is $\mathsf{DP}$-complete. 

\item Standard, e.\,g., \cite[Ex.~3.8.1]{Downey}.

\item Follows from the fact that \textsc{Hitting Set} is $\mathsf{W[2]}$-complete \cite[Cor.~23.2.2(5)]{Downey} (for which Downey \& Fellows credit Wareham).

\item Follows from the same result for \textsc{Hitting Set} \cite[Thm.~1.1(2)]{CyganEtAl}. 

\end{enumerate}
This completes the proof.
\end{proof}

\subsection{Point of No Information}
The point of no information is the point at which every $k$-window has all $2^k$ distinct subsequences. This problem is in $\mathsf{coNP}$ because if a given $k$ is not the point of no information, you can provide the window and a missing subsequence as witness, which could be checked by checking each of the $m$ subsequences. 

\subsection{Perfect Reconstruction}
The decision version of \textsc{Perfect Reconstruction}---where $k$ is part of the input and the output decides whether the point of perfect reconstruction is at most $k$, or equivalently whether $S = \Recon_k(S)$---appears to lie somewhere between the first and second levels of the polynomial hierarchy, and we make this more precise in this section.

There is no known quick guaranteed verification that a given $k$ is the point of perfect reconstruction, because this would require confirming that every possible string is either in the data set or has some window of $k$ indices that is not observed in our data set. The negative is similarly difficult to prove; to verify that the point of perfect reconstruction is greater than $k$, we would need to provide at least one string that does not have a window of size $k$ missing from our set, which would require proving there is no hitting set of that size. 

There is a higher-level version of \textsc{Hitting Set} that matches the \textsc{Perfect Reconstruction} problem, built using the \textsc{String Non-Containment} reduction of Thm.~\ref{SCtoHS}. In the following definition, given two sets $S,X$, we define $S \oplus X$ to be their symmetric difference: $S \oplus X := (S \backslash X) \cup (X \backslash S)$.

\begin{quotation}
\noindent \textsc{Max-Min Toggled Hitting Set}

\noindent \textit{Input}: A collection of subsets $S_1, \dotsc, S_m$ of $\{1,\dotsc,n\}$, and a natural number $k$

\noindent \textit{Decide}: For all $X \subseteq \{1,\dotsc,n\}$ such that $X \notin \{S_1, \dotsc, S_m\}$, does the collection of sets $S_1 \oplus X, S_2 \oplus X, \dotsc, S_m \oplus X$ have a hitting set of size at most $k$?
\end{quotation}

\begin{corollary}
The \textsc{Perfect Reconstruction} problem is equivalent to the \textsc{Max-Min Toggled Hitting Set} problem.
\end{corollary}

In fact, the two problems are ``essentially the same'' in the same manner as in Thm.~\ref{SCtoHS}, but because we do not know about (in)approximability or completeness of either \textsc{Perfect Reconstruction} or \textsc{Max-Min Toggled Hitting Set}, we cannot make so thorough a statement about their complexity.

\begin{proof}
We use the reductions from the proof of Thm.~\ref{SCtoHS}. If $(S,k)$ is an instance of the \textsc{Perfect Reconstruction} problem, we get an instance of \textsc{Max-Min Toggled Hitting Set} where the $i$-th string in $S$ gets treated as an indicator function of a subset $S_i \subseteq [n]$; the collection of all such $S_i$ is our instance of \textsc{Max-Min Toggled Hitting Set}. Then for a string $x \in \{0,1\}^n$, let $X \subseteq [n]$ be the set of which $x$ is its indicator function. Then by the proof of Thm.~\ref{SCtoHS}, we have $x \notin \Recon_k(S)$ if and only if the $X$-toggled version of the \textsc{Hitting Set} instance above has a hitting set of size at most $k$. And thus $S = \Recon_k(S)$ iff $(S_1,\dotsc,S_m; k)$ is a yes-instance of the \textsc{Max-Min Toggled Hitting Set} problem.
\end{proof}

This is very similar to the structure of min-max problems proven to be $\mathsf{\Pi_2 P}$-complete \cite{Ko1995}. However there is one critical difference that leads us to believe this problem might not actually be hard for the second level of $\mathsf{PH}$. The \textsc{Perfect Reconstruction} problem (or \textsc{Max-Min Toggled Hitting Set}) can be fully defined by a single instance of \textsc{Hitting Set}, and all other subproblems are closely related to the original instance as above---they are merely ``toggled'' versions of it---whereas in other problems that are hard for the second level of the polynomial hierarchy that we are aware of, the subproblems for differential values of the universally quantified parameter can be unrelated. 

However, the complexity class $\mathsf{\Pi_2 P}[\star k, 2]$, introduced by de Haan and Szeider in the context of multi-parameter complexity \cite{dHS,paraCompendium19}, seems like a closer fit for \textsc{Perfect Reconstruction}. Rather than define this class intrinsically, we define it by one of its standard complete problems: An instance is a quantified Boolean formula of the form $\forall x \exists y \varphi$ where $\varphi$ has weft 2,\footnote{We recall the definition of weft from, e.\,g., \cite{Downey}: an infinite set of Boolean formulas $\varphi_1, \varphi_2, \dotsc$ has weft at most 2 if there is a constant $c$ such that for all $i$, on any path from a leaf to the root of $\varphi_i$, there are at most 2 gates of fan-in greater than $c$. Or, in standard if somewhat less-precise parlance, $\varphi = (\varphi_i)_{i=1,2,3,\dotsc}$ has weft at most 2 if on any leaf-to-root path there are at most 2 gates of unbounded fan-in.} and an integer $k$ (in the multi-parameter setting, $k$ is standardly used as the parameter as well), and the question is to decide whether for every assignment to $x$, there is an assignment to $y$ \emph{of Hamming weight $k$}, such that $\varphi(x,y)=1$. 

\begin{proposition} \label{thm:Pi2P}
Testing whether a given set of strings is perfectly reconstructible from its $k$-way projections is in $\mathsf{\Pi_2 P}[\star k, 2]$.
\end{proposition}

We leave open the question of completeness:

\begin{open} \label{open:perfect}
Is testing whether $S = \Recon_k(S)$ complete for $\mathsf{\Pi_2 P}[\star k, 2]$ under fpt reductions?
\end{open} 

A positive answer would imply that there is no fpt-reduction from \textsc{Perfect Reconstruction} to \textsc{UNSAT} unless there is a subexponential reduction from $\exists \forall \textsc{DNF-SAT}$ to \textsc{SAT} \cite[Cor.~40]{dHS}.

\begin{proof}[Proof \if\full\else sketch \fi of Prop.~\ref{thm:Pi2P}]
We can naturally translate a given instance of our problem by considering $x$ to be a particular string we are checking, and $y$ the chosen $k$ indices to confirm whether that string is in our reconstructed set or not. If $S = \Recon_k(S)$, then for all strings $x$ there must be some choice of indices where $x$ disagrees with every  string in $S$ (note that quantifying over all strings in the original data set can be done in polynomial time in the input size, so it does not ``count'' as an additional unbounded quantifier), OR the string $x$ must be in $S$. We can use this to build a boolean formula in a natural way, and observe that that construction produces a boolean formula of weft 2.

\if\full
For a given string in our data set $s_i \in S$, we construct a formula $c_i$ for each $i$ as follows: If index 1 is chosen AND our random string disagrees with $s_i$ at index 1, OR index 2 is chosen AND our random string disagrees with $s_i$ at index 2, OR ... OR index $n$ is chosen AND our random string disagrees with $s_i$ at index $n$. If $b \in \{0,1\}$, let $x_i^b = x_i$ if $b=0$ and $\neg x_i$ if $b=1$. Then our clauses are: 

$$c_i = \left( y_1 \land x_1^{s_{i,1}} \right) 
    \lor 
    ( y_2 \land x_2^{s_{i,2}} )
    \lor \dots \lor \left( y_n \land x_n^{s_{i,n}} \right)
$$

Our final quantified formula is then:
$$
(\forall x)(\exists y) [(wt(y)=k) \land \left( \bigwedge_{i=1}^m c_i \right) \lor \bigvee_{i=1}^m \bigwedge_{j=1}^n (\neg x_j^{s_{i,j}})]
$$
The formula is clearly in $\mathsf{\Pi_2 P}[\star k]$; to see that it is in $\mathsf{\Pi_2 P}[\star k, 2]$, observe that the formula involves at most 2 nested unbounded quantifiers, so it has weft 2. 
\fi
This completes the proof. \qedhere
\end{proof}

\if\full
\begin{example}
    Let's say we have a data set with three strings of length four: {0011, 1000, 0101}. This would translate to the following boolean equation:

    \begin{align*}
\forall X \exists Y_k & \left[ ((y_1 \land x_1) \lor (y_2 \land x_2) \lor (y_3 \land \neg x_3) \lor (y_4 \land \neg x_4)) \land \right. & \text{(0011)} \\
& ((y_1 \land \neg x_1) \lor (y_2 \land x_2) \lor (y_3 \land x_3) \lor (y_4 \land x_4)) \land &  \text{(1000)} \\
& \left. ((y_1 \land x_1) \lor (y_2 \land \neg x_2) \lor (y_3 \land x_3) \lor (y_4 \land \neg x_4)) \right] \lor & \text{(0101)} \\
& (\neg x_1 \land \neg x_2 \land x_3 \land x_4) \lor & \text{(0011)}\\
& (x_1 \land \neg x_2 \land \neg x_3 \land \neg x_4) \lor & \text{(1000)} \\
& (\neg x_1 \land x_2 \land \neg x_3 \land x_4) & \text{(0101)}
\end{align*}
\end{example}
\fi

\begin{observation} \label{obs:perfect1}
Testing whether $1$ is the point of perfect reconstruction of a set of $m$ strings of length $n$ can be done in linear time $O(nm)$ on a RAM. 
\end{observation}

The proof mostly involves a careful accounting of the arithmetic operations needed to verify the condition from Obs.~\ref{obs:1recon}.

\begin{proof}
The point of perfect reconstruction is $1$ if and only if the set is the direct product of its 1-way projections. To test whether this is the case, suppose $n_i$ distinct letters occur in index $i$. Then $S$ is 1-reconstructible iff $|S| = \prod n_i$. To see that this can be computed in linear time, suppose there are $m$ strings of length $n$ over an alphabet of size $a$, so the input size $\Theta(nm \log a)$. For each index $i=1,\dotsc,n$, the algorithm goes through all $m$ strings and keeps track of which letters it has seen so far, in a bit-array of length $a$. Every time it updates the array $a$ by changing some entry from $0$ to $1$, it also updates the value of $n_i$. As $n_i \leq a$, it is an integer of at most $\log a$ many bits, that gets incremented by $1$ at most $\min\{a,m\}$ times, which happens for each of the $n$ indices, so this all takes time at most $O(nm \log a)$, which is linear in the input size. Finally, we multiply the $n_i$ one at a time, and if the value ever exceeds $|S|$ we stop and answer No. Thus the bit-length we need to keep for all of these is never more than $\log m$, and we perform at most $n$ multiplications, for a total cost of $O(n (\log m)^2) \leq O(nm)$, which is again linear in the input size. This completes the proof.
\end{proof}

\begin{theorem} \label{thm:2recon}
Over an alphabet of size 2, testing whether $2$ is the point of perfect reconstruction can be done in time $O(n^2 (m + \log n))$ for sets of $m$ strings of length $n$.
\end{theorem}

The proof here is more involved: it first involves reducing to 2-SAT, and then using Feder's \cite{Feder} algorithm for enumerating 2-SAT solutions with polynomial delay as a subroutine. 

\begin{proof}
We first check whether 1 is already the point of perfect reconstruction, by Observation~\ref{obs:perfect1}. If not, then we reduce to the question of enumerating solutions to a 2-SAT instance, which is then solvable by an algorithm of Feder \cite{Feder}. 

Let $S$ be the given set of $m$ strings of length $n$. The 2-SAT instance $\varphi$ will have $n$ variables $x_1, \dotsc, x_n$. For each pair of distinct indices $i,j \in [n]$, we add clauses on the variables $x_i, x_j$ such that the set of assignments to $x_i,x_j$ satisfying those clauses are precisely the pairs of bits occurring in $S|_{\{i,j\}}$. (If $S|_{\{i,j\}}$ includes all 4 possibilities, we do not add any clauses on the pair $x_i, x_j$.) Thus, the set of satisfying assignments to $\varphi$ is identical to the 2-reconstruction of $S$. Building $\varphi$ takes $O(n^2 (m + \log n))$ time (the $\log n$ is because it takes logarithmic time to even write down a single index $i \in [n]$). As $\varphi$ has $O(n^2)$ clauses, and each variable occurs in at most $O(n)$ clauses, Feder's algorithm \cite{Feder} enumerates solutions to $\varphi$ with $O(n^2)$ preprocessing time and $O(n)$ time per solution enumerated. We run Feder's algorithm, keeping a count of how many solutions it has found so far. If and when it enumerates solution number $m+1$, we can stop and report that $S$ is not 2-reconstructible, for the 2-reconstruction necessarily contains a string not in $S$. Conversely, since we know $S$ is contained in its 2-reconstruction, if the enumeration algorithm does not find $\geq m+1$ solutions, then it will enumerate precisely the strings in $S$, and thus $S$ is 2-reconstructible. The total time for this procedure is then $O(n^2 + nm)$, but this is already dwarfed by the time to produce $\varphi$.
\end{proof}

We leave as an open question:

\begin{open} \label{open:smallK}
What is the complexity of deciding whether a given set of strings over an alphabet of fixed size $\geq 3$ is perfectly reconstructible from its 2-way projections? In particular, it is in $\mathsf{P}$? Is it $\mathsf{coNP}$-complete?
\end{open}

As a polynomial-delay enumeration algorithm for binary CSPs (each constraint is on at most 2 variables) over an alphabet of size $a \geq 3$ would be able to solve the $a$-coloring problem on graphs, no such algorithm exists unless $\mathsf{P} = \mathsf{NP}$. Thus, while it is conceivable that Thm.~\ref{thm:2recon} could be extended to larger alphabets, doing so requires a different approach than in our current proof. It is also possible that the problem is $\mathsf{coNP}$-complete, though we ran into difficulties with at least one attempted proof.\footnote{The following natural attempt to show that testing 2-reconstructibility over an alphabet of size $\geq 3$ is $\mathsf{coNP}$-hard runs into an obstacle that we record here. Suppose we wish to reduce from the \textsc{Another Solution to Graph 3-Colorability}: given a graph $G$ and a 3-coloring $c$, decide whether there is another proper coloring of $G$, that does not just differ from $c$ by permuting the colors. As \textsc{Another Solution to 3-SAT} is $\mathsf{ASP}$-complete (see \cite{UedaNagao,YatoSeta}), hence $\mathsf{NP}$-complete, and the standard reduction from 3-SAT to 3-Coloring is nearly parsimonious---it simply multiplies the number of solutions by $3!=6$ because of relabeling the colors---\textsc{Another Solution to Graph 3-Colorability} is also $\mathsf{ASP}$-complete, hence $\mathsf{NP}$-complete. We may try to reduce this to 2-reconstructibility for an alphabet of size 3 as follows: given an $n$-vertex graph $G$ and a proper 3-coloring $c\colon G \to \{r,g,b\}$, we build a set of 6 strings of length $n$ over the alphabet $\{r,g,b\}$, where the indices correspond to the vertices of $G$, and the strings are the 6 colorings gotten from $c$ by permuting the 3 colors in all possible ways. Now, for any 2-window corresponding to an edge of $G$, among the 6 strings we see all possible pairs of distinct colors, giving the correct constraint on those two indices, namely, that they be distinct. However, for two vertices $u,v$ such that $c(u)=c(v)$, all 6 strings will have the the same value in the indices corresponding to $u,v$. Thus the 2-window corresponding to $u,v$ enforces that, in any string in the 2-reconstruction, the colors of $u,v$ must be equal. And thus this set is automatically 2-reconstructible, it does not allow any of the other potential 3-colorings of $G$.}

\section{Overlap graph algorithm} \label{sec:algorithm}
In \S\ref{sec:overlap_intro} we introduced the concept of an overlap graph as adapted to fit our problem. Here we explain our efficient algorithms for constructing and using the overlap graph. 

In order to construct our overlap graph, we must first choose an ordering of our indices. Since we can use all subsequences of length $k$, including those with non-consecutive indices, this allows us to reorder our indices with no loss of information. We discuss our current methods for choosing an ordering in \S\ref{sec:preprocessing}; this is an area where further improvements may be possible.

After choosing our ordering, we will then create rows in our overlap graph using consecutive substrings that overlap by $k-1$. Our first row will contain all substrings from our data set observed at indices 1 through $k$; the second row will contain all substrings at indices 2 through $k+1$ and so on, cycling around, viz. the last row corresponds to indices $\{n, 1, 2, \dotsc, k-1\}$. After populating all rows of our overlap graph we will draw edges from the first row to the second row where indices 2 through $k$ are identical; we will similarly do the same for all adjacent rows in the overlap graph for their $k-1$ length overlaps until we eventually draw edges from the final row of the graph back to the first row. 

To find the edges between row $i$ and $i+1$: for each $k$-mer in row $i$, we take its length-$(k-1)$ suffix $s$, and then check whether $s0$ (resp., $s1$) are present in row $i+1$. If the maximum number of $k$-mers in a row is $t$, this takes at most $O(t \log t)$ time if the strings in a given row are stored in sorted order, so that they can be found by binary search. With $m$ strings of length $n$, the overall time to construct this graph is thus $O(nmk + t \log t) \leq O(nmk + k 2^k)$. 

The strings that are consistent with these $n$ windows of size $k$ are precisely the length-$n$ cycles in the overlap graph, so our next step is to compute the set of all such cycles. Because of the layered structure of our overlap graph, each such cycle must contain exactly one node from each row, which we take advantage of in order to find all such cycles quickly. 
Finding these cycles can be done using modifications of breadth-first or depth-first search; one simple way to do so is, for each vertex $v$ in the first layer, use BFS to depth $n$ to find all directed paths of length exactly $n$ back to $v$. 

After we have reconstructed these strings, we separate out the strings that were not in our initial data set and use a \textsc{Hitting Set} subroutine to determine which of them are in the $k$-reconstruction and, therefore, whether or not the Point of Perfect Reconstruction is at most $k$. For our experimental evaluation, even with a brute force \textsc{Hitting Set} algorithm our overlap graph method demonstrates significant speedup over a simply greedy method once $n$ is large enough (roughly $n \geq 8$).

\subsection{Ordering the columns} \label{sec:preprocessing}

Since we have access to all subsequences, including those without consecutive indices, we can pick any column ordering to create an overlap graph. As such, we can perform a preprocessing step to reorganize the columns in an attempt to make our overlap graph more efficient.

In order to optimize our overlap graph, we want to minimize the size and minimize the number of cycles. A strategy towards both is to arrange the matrix so that the most similar columns are adjacent to one another. By placing similar columns next to each other we minimize the number of substrings, which in turn will minimize the size of the overlap graph and therefore the number of cycles. 
    
While there are several methods to accomplish this goal efficiently, the one we have chosen both for simplicity and efficiency is this: rewrite all binary strings in a $\{-1, 1\}$ basis, and to then calculate the dot products of each pair of columns. However, note that if we had a column of all 1's and a column of all -1's, we would similarly want them to be beside each other; in other words, it does not matter if the agreements between two columns are $-1=1$ or $1=1$, so long as they are consistent. As such, we will take the absolute value of these dot products. 
    
We can calculate all such values for each pairwise combination of columns and store these values in an $n \times n$ matrix. We then permute this matrix such that the sum one value off the diagonal is maximized using a simple greedy algorithm. These values correspond with the dot products between column $i$ and column $i+1$, which we will then use to define the index ordering of our strings for the overlap graph calculation. 

\subsection{Observations on overlap graphs}

There are several additional features of our overlap graph that we can take advantage of for further speed-ups in practice.

Every node in the overlap graph must be from an original string, so it must belong to at least one $n$-cycle. Furthermore, if it belongs to only one $n$-cycle, that cycle must have been one of our original strings. In that case, if what we're interested in is the point of perfect reconstruction, we can remove that node from the overlap graph before looking for more cycles, since that node does not contribute to any cycles that correspond to strings in the $k$-reconstruction that aren't in the original input. Note that we can determine the number of cycles a node belongs to much faster than we can list all such cycles: if $A$ is the adjacency matrix of the overlap graph, then $(A^n)_{vv}$ is the number of cycles through vertex $v$. And $A^n$ can be calculated by repeated squaring using only $O(\log n)$ matrix products.

Note that, since each layer consists of $k$-mers, it has at most $2^k$ nodes, and there are exactly $n$ layers. Furthermore, since every node has out-degree either 1 or 2, we get $|V| \leq n 2^k$ and $|V| \leq |E| \leq 2|V|$, so the graph is quite sparse. 

Fig.~\ref{fig:OLGtraits} gives an example of an overlap graph, its adjacency matrix $A$, and $A^n$. As can be seen in the figure, the layered structure of the overlap graph corresponds to its adjacency matrix consisting of precisely $n$ blocks, each of which is at most $2^k \times 2^k$, and where the blocks appear in a cyclic permutation pattern. Thus, even when we take powers of $A$---which need no longer be so sparse---it still consists of precisely $n$ such blocks of size at most $2^k \times 2^k$. So every power of $A$ has at most $4^k n$ nonzero entries, but $A$ could be as big as $2^k n \times 2^k n$, so the number of nonzero entries in any power of $A$ is always a factor of $\leq 1/n$ times the total number of entries in the matrix. The known block structure can also be used directly to compute the powers of $A$.  

Another valuable observation is that in the case $k=n-1$, we use all possible subsequences in the overlap graph and therefore use all information available to us. In this case, the strings produced by our overlap graph precisely correspond with the strings produced by the $k$-reconstruction of our data set, and no additional algorithmic work is required. Therefore, in this particular case, the \textsc{Perfect Reconstruction} problem is in $\mathsf{P}$. Note that this result does not even extend to $k=n-2$, for already in that case some windows are omitted from the overlap graph, and our strategy will still need to use \textsc{Hitting Set} on any strings that aren't in the original data set to determine whether or not they are in the $k$-reconstruction.

\begin{figure}
    \centering
    \begin{subfigure}[b]{0.49\textwidth}
    \pdftooltip{\includegraphics[width=\textwidth]{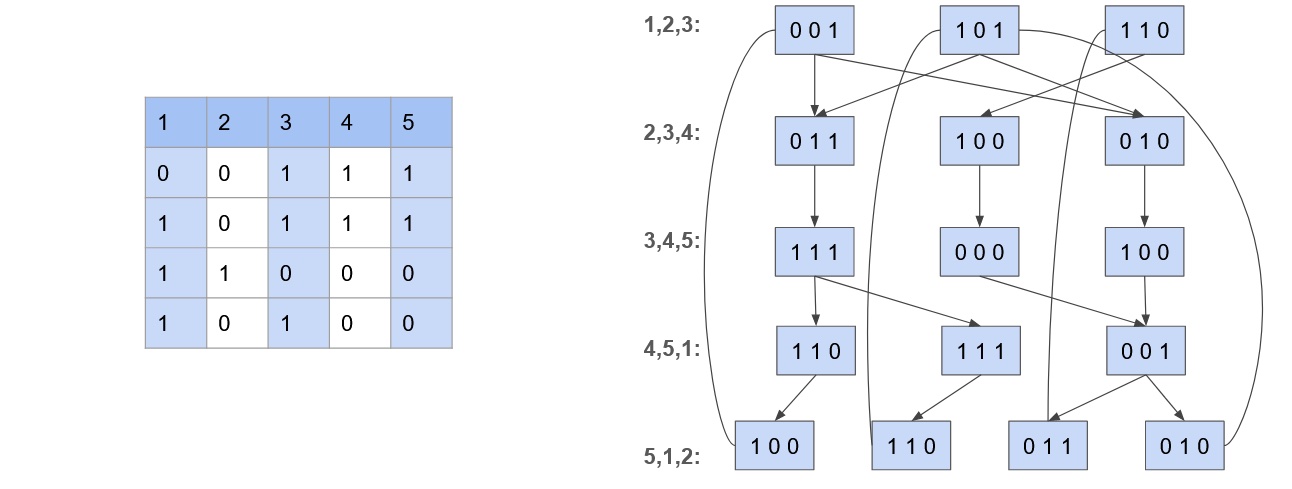}}{A table displaying the four strings of length 5 as its rows: $00111$, $10111$, $11000$, $10100$, next to the corresponding overlap graph. The overlap graph has three nodes in its first four layers, and four nodes in its final layer. Edges go between layers, including looping up from the final layer back to the first layer.}
    \caption{A data set consisting of 4 strings each of length 5 with the indices enumerated 1--5, and its corresponding overlap graph at $k=3$.}
    \end{subfigure}
\begin{subfigure}[b]{0.49\textwidth}
        \pdftooltip{\includegraphics[width=\textwidth]{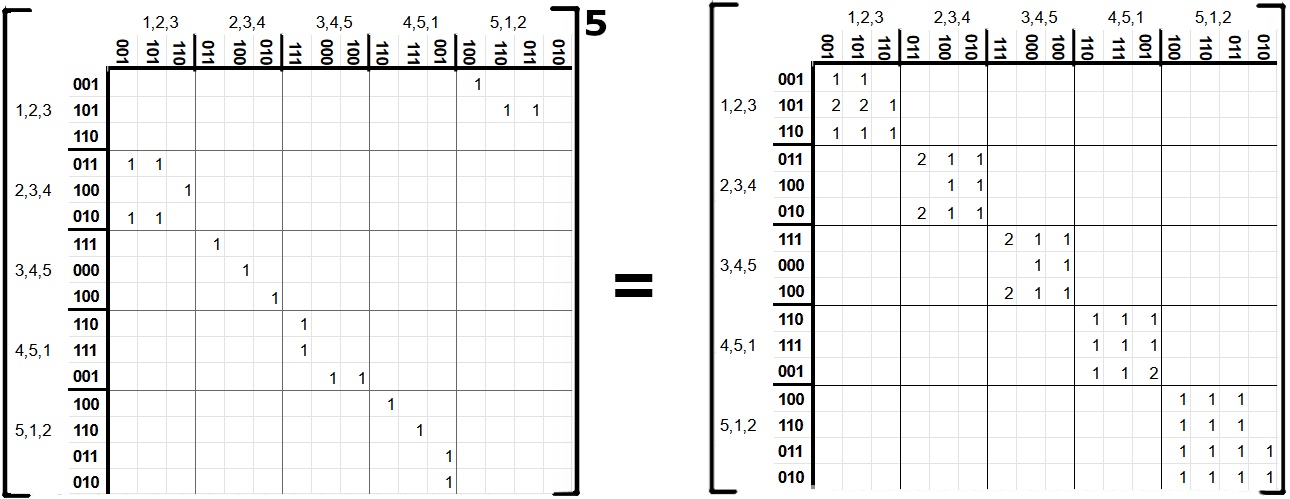}}{The adjacency matrix of the overlap graph from the other half of the figure. It is a 19-by-19 matrix, arranged into block sub-matrices corresponding to the layers of the overlap graph. The nonzero blocks are those just below the diagonal blocks. Since the overlap graph has $5$ layers, when raised to the fifth power, we get a block-diagonal matrix with five blocks.}
        \caption{The adjacency matrix $A$ of the overlap graph. Nodes are labeled by substrings; blocks are labeled by the indices they originated from in our data set. }
    \end{subfigure}

    \caption{An example data set, its overlap graph, and its adjacency matrix $A$. The diagonal entries of the matrix $A^5$ count the number of 5-cycles through each vertex. In the lower-right block we can see that each vertex in that layer has exactly 1 cycle through it, so those cycles must precisely be the original strings we started with. Thus in this case the diagonal entries of $A^n$ imply that this set of strings is perfectly reconstructed at $k=3$.     \label{fig:OLGtraits}}
\end{figure}

\section{Experimental evaluation}
We ran a series of experiments to calculate the $k$-reconstruction of a set, comparing our overlap graph algorithm (including all needed calls to \textsc{Hitting Set}, which we solve using a simple brute-force algorithm) against a simple greedy algorithm. The greedy algorithm we consider builds up the strings starting from a size-$k$ window, ensuring compatibility with the data set each time it extends the strings by a single index. Although simple, this is still significantly faster than brute-forcing over all $2^n$ strings; see App.~\ref{app:greedy} for details of the greedy algorithm. 

All experiments were conducted using C++ on a machine with 32.0 GB of RAM and a 12th Gen Intel(R) Core(TM) i7-12700K processor, compiled using \texttt{g++} and with the C++17 standard (\texttt{-std=c++17}) and no additional optimization options. Each experiment began by randomly generating a data set and then testing the time to completion for our overlap graph-based algorithm, and then the time to completion for the greedy algorithm. Each combination of parameters $m$, $n$, and $k$ shown were evaluated over 30 trials, except for those with $m \geq 500$, for which we evaluated over 10 trials. The code we used, as well as the randomly generated data sets that were generated, can be found at \cite{code}.

Although exhaustive exploration of all combinations of $(n,m,k)$ was infeasible for sufficiently large $n$, we performed experiments showing the scaling as we fix two of the parameters and let the third vary. At $n=12$, we cover a large portion of entire parameter region of $(m,k)$: we did experiments for all $2 \leq k \leq n$ and $m$ up to $2000$ (the maximum at $n=12$ being $4096$), giving us a good overall picture of the performance of the two algorithms. In the rest of the section we highlight this data and give theoretical explanations that match the experimentally observed runtime behavior.

We start with the behavior as the string length $n$ grows: in Fig.~\ref{fig:compareN} we see that our overlap graph algorithm appears exponentially better than the greedy algorithm as a function of $n$, the length of the strings. In Fig.~\ref{fig:big} we see that our algorithm out-performs the greedy algorithm for most values of $(m,k)$, with very small $(m,k)$ being some of the exceptions. Note that in the latter figure we have $n=12$, we look at all values of $k$, and $m$ up to $2000 \approx 2^n / 2$, giving a fairly comprehensive comparison of the two algorithms.

\begin{figure}[!htbp]
    \begin{subfigure}[b]{0.49\textwidth}
    \pdftooltip{\includegraphics[width=\textwidth]{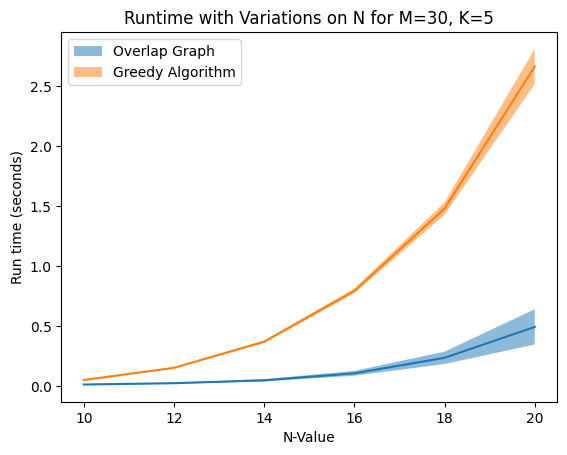}}{A plot showing the runtime of our overlap graph method and greedy method, as a function of the string length $n$. Both display  exponential growth, with the greedy algorithm appearing to grow exponentially with a larger base of the exponent. The overlap graph runtimes range from 0.011s at $n=10$ to 0.49s at $n=20$, while the greedy algorithm runtimes range from 0.049s at $n=10$ to 2.66s at $n=20$. See Appendix~\ref{app:data} for the full data table and our git repo for the data in CSV format.}
    \caption{At $m=30$ strings and window size $k=5$, our overlap graph algorithm is faster than the greedy algorithm for all lengths $n \geq 10$.    \label{fig:compareN}
}
    \end{subfigure}
    \begin{subfigure}[b]{0.49\textwidth}
    \pdftooltip{\includegraphics[width=\textwidth]{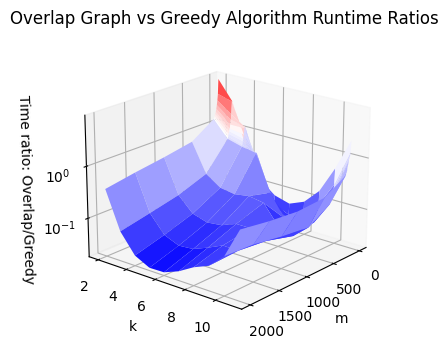}}{A 3D surface plot of the ratio in runtimes between the overlap graph algorithm and greedy algorithm for $n=12$, with the base axes being $k$ (ranging from $2$ to $11=n-1$) and $m$ ranging from $0$ to $2000$. When $k \leq 4$ and $m \leq 100$ there is a small region where the greedy algorithm is faster by a factor less than $10$, and the rest appears to be concave up but well below a factor of 1, meaning that the overlap graph algorithm is faster. See Appendix~\ref{app:data} for the full data table and our git repo for the data in CSV format.}
        \caption{At $n=12$, our algorithm is faster for all but a very small region of $(m,k)$ parameter values. See \cite{code} for an animated version for better viewing.}
    \label{fig:big}
    \end{subfigure}
\end{figure}

In Fig.~\ref{fig:zoom}(a) we zoom in on the region $m \leq 100$; comparing with Fig~\ref{fig:zoom}(b) we see that the region where the greedy algorithm is faster is contained within the region where $k$ is at most (or very close to) the point of no information. In random sets of strings, by the coupon collector problem, this region (in Fig~\ref{fig:zoom}(b)) is where $m \geq k2^k$. To see that it really is how close we are to the point of no information that is governing when the greedy algorithm is faster for small values of $k$, at $n=12$, we explored a 2D region of $(m,k)$ parameter space that includes all values of $k$ (Fig.~\ref{fig:zoom}(b)). 

\begin{figure}[!htbp]
    \centering
        \begin{subfigure}[b]{0.49\textwidth}
    \pdftooltip{\includegraphics[width=\textwidth]{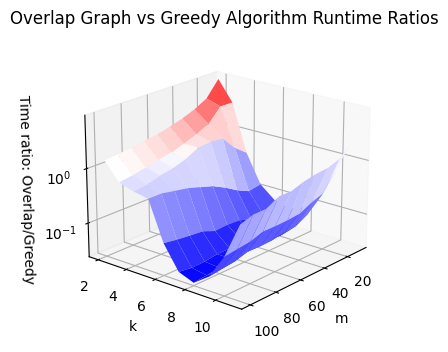}}{A 3D surface plot of the ratio in runtimes between the overlap graph algorithm and greedy algorithm for $n=12$, with the base axes being $k$ (ranging from $2$ to $11=n-1$) and $m$ ranging from $0$ to $100$. When $k \leq 4$ there is a strip in the plot where the greedy algorithm is faster by a factor less than $10$, and the rest appears to be concave up but well below a factor of 1, meaning that the overlap graph algorithm is faster. See Appendix~\ref{app:data} for the full data table and our git repo for the data in CSV format.}
    \caption{Ratio of runtimes between our overlap graph algorithm and the greedy algorithm at $n=12$, for $m \leq 100$ and all $k$.}
    \end{subfigure}
    \begin{subfigure}[b]{0.49\textwidth}
\pdftooltip{\includegraphics[width=\textwidth]{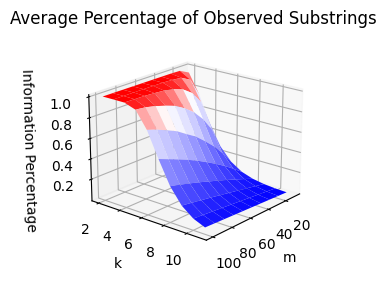}}{A 3D surface plot of the fraction of strings seen in each $k$-window. In the region $m \gtrsim k 2^k$ the fraction is very close to one, and smoothly decreases to $0$ as one moves away from that region in the direction of increasing $k$. The region in part (a) of the figure where the greedy algorithm is faster is contained in the region of this plot where the fraction is close to 1.}
    \caption{Average fraction of $2^k$ subsequences seen in each $k$-window. (1.0 is the same as having no information, $\Recon_k(S) = \{0,1\}^n$.)}
    \end{subfigure}

    \caption{Zoomed in version of Fig.~\ref{fig:big}, focusing on $m \leq 100$. At small $k$ and for some $m \gtrsim k 2^k$, the greedy algorithm outperforms our overlap graph algorithm because $k$ is below or close to the point of no information.}
    \label{fig:zoom}
\end{figure}

However, the region where the greedy algorithm is faster is only a subset of the region of no information; we now explore this phenomenon in more depth. At $n=12$ (as in the previous figures), at small values of $k$ we see that there is a change between, say, $m=40$ where the greedy algorithm is faster up to $k=4$, while once   $m \geq 90$, the greedy algorithm is only faster at $k=2$. In Fig.~\ref{fig:compareK}(a) and (b) we look at these two slices of the 3D picture (that is, fix $n=12$ and $m$ either 40 or 90, and look at the variation with $k$). 
On those log-plots, we see that both algorithms' runtimes peak at some relatively small value of $k$; after that, the overlap graph algorithm's runtime continues to decrease with increasing $k$, while that of the greedy algorithm seems to follow a downward quadratic peaking at $k=7=(n/2)+1$, corresponding to the binomial coefficient $\binom{n}{k-1}$. The following lemma partially explains the presence of the binomial coefficient $\binom{n}{k-1}$: 

\begin{figure}[!htbp]
    \centering
    \begin{subfigure}[b]{0.49\textwidth}
    \pdftooltip{\includegraphics[width=\textwidth]{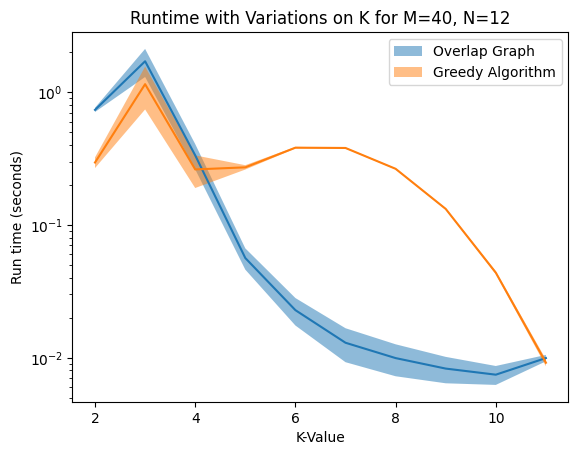}}{A log-plot of runtimes for $n=12, m=90$ as $k$ varies. At $k \leq 4$, both algorithms peak at $k=3$ with runtimes close to one another, though the greedy algorithm slightly faster. As $k$ increases, the overlap graph continues to decrease until $k=11$ when it increases slightly. In contrast, the greedy algorithm takes significantly more time, which, for $k \geq 4$, displays a negative quadratic behavior that peaks around $k=7$ (corresponding to the Gaussian approximation to the binomial coefficient). See Appendix~\ref{app:data} for the full data table and our git repo for the data in CSV format.}
    \end{subfigure}
\begin{subfigure}[b]{0.49\textwidth}
    \pdftooltip{\includegraphics[width=\textwidth]{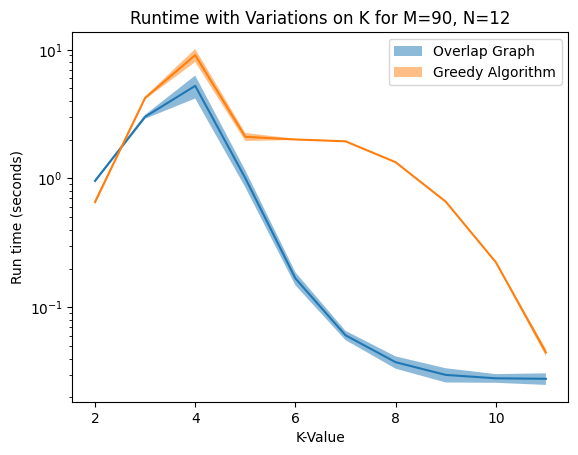}}{A log-plot of runtimes for $n=12, m=90$ as $k$ varies. Very similar to the $m=40$ plot, except for $k \geq 3$ our overlap graph algorithm is faster. At $k \leq 4$, both algorithms peak at $k=3$ with runtimes close to one another. As $k$ increases, the overlap graph continues to decrease until $k=11$ when it increases slightly. In contrast, the greedy algorithm takes significantly more time, which, for $k \geq 4$, displays a negative quadratic behavior that peaks around $k=7$ (corresponding to the Gaussian approximation to the binomial coefficient). See Appendix~\ref{app:data} for the full data table and our git repo for the data in CSV format.}
    \end{subfigure}

    \begin{subfigure}[b]{0.49\textwidth}
    \pdftooltip{\includegraphics[width=\textwidth]{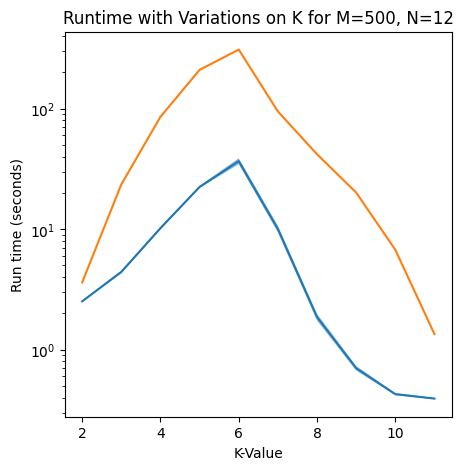}}{A log-plot of runtimes for $n=12, m=500$ as $k$ varies. Both algorithm display more of a ``mountain peak'' shape with some slight curvature, but a sharp peak at $k=6$, with the overlap graph algorithm consistently faster by almost a factor of $10$. See Appendix~\ref{app:data} for the full data table and our git repo for the data in CSV format.}
    \end{subfigure}
\begin{subfigure}[b]{0.49\textwidth}
    \pdftooltip{\includegraphics[width=\textwidth]{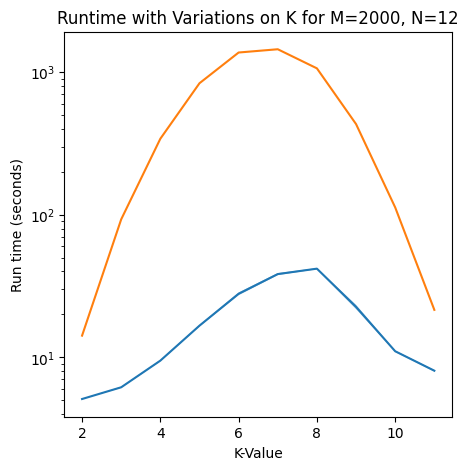}}{A log-plot of runtimes for $n=12, m=2000$ as $k$ varies. The greedy algorithm displays a negative quadratic-looking behavior (seemingly corresponding to the binomial coefficient), peaking at $k=7$ with runtime $> 1000s$. The overlap graph algorithm also displays a single-peaked behavior, but much lower (maximum less than $100s$ at $k=8$) and with peak slightly skewed to the right compared to the greedy algorithm runtime. See Appendix~\ref{app:data} for the full data table and our git repo for the data in CSV format.}
    \end{subfigure}
    
    \caption{Comparing runtime of our overlap graph algorithm with the greedy algorithm on random data sets of length-12 strings ($n=12$), for all values of the window size $k$ (between $2$ and $11=n-1$), at (top left) $m=40$ strings, (top right) $m=90$, (bottom left) $m=500$, (bottom right) $m=2000$.}
    \label{fig:compareK}
\end{figure}

\begin{lemma} \label{lem:greedy}
When $k$ is at most the point of no information (i.e., $\Recon_k(S) = \{0,1\}^n$), the greedy algorithm has runtime $\Theta(\binom{n}{k-1}2^n)$.
\end{lemma}

In fact, as we'll see in the proof, our lower and upper bounds on the runtime are within a factor of 2 of each other.

\begin{proof}
If $k$ is at most the point of no information, then $m \geq 2^k$, and the runtime can be estimated as follows. In the first window, the greedy algorithm sees all $2^k$ possible strings. For each of those and each possible extension, $2^{k+1}$ in all, it must check $\binom{k}{k-1}$ windows. Since $\Recon_k(S) = \{0,1\}^n$, at the $i$-th stage, it will have found $2^{k+i}$ strings, and must check $2^{k+i+1}$ strings against $\binom{k+i}{k-1}$ windows. For a total runtime governed primarily by: 
\begin{align*}
\sum_{i=0}^{n-k} 2^{k+i+1} \cdot \binom{k+i}{k-1} = \sum_{j=k}^n 2^{j} \binom{j}{k-1}
\end{align*}
As $j \leq n$ for all terms in the latter sum, we can bound the sum by 
\[
\leq \binom{n}{k-1} \sum_{j=k}^n 2^j = \binom{n}{k-1}(2^{n+1}-2^{k+1}) \le \binom{n}{k-1}2^{n+1}.
\]
This is within a factor of 2 of the exact value, since the largest single summand is already $\binom{n}{k-1}2^n$.
\end{proof}

Lemma~\ref{lem:greedy} and Fig.~\ref{fig:compareK} suggest examining what these plots look like when normalized by $\binom{n}{k-1}$; see Fig.~\ref{fig:compareKNormalized}. In that figure we have also plotted the number of extra reconstructed strings (beyond the original data set $S$), so that one may easily tell where the points of no information and perfect information are: no information is when the number of extra reconstructed strings is maximal, and perfect information is when it is 0. As we see in Fig.~\ref{fig:compareKNormalized}, after dividing out by $\binom{n}{k-1}$, much of the behavior of the runtimes is governed by the number of extra reconstructed strings. At $m=2000$, we see what appears to be behavior \emph{very} closely governed by $\binom{n}{k-1}$ in Fig.~\ref{fig:compareK}(d); when normalizing by $\binom{n}{k-1}$ (Fig.~\ref{fig:compareKNormalized}(d)) we see there is still quite a bit of variation, but that variation is mostly within a multiplicative factor of 2, which would only appear as a tiny additive correction on the log-plot of Fig.~\ref{fig:compareK}(d).

\begin{figure}
    \centering
    \begin{subfigure}[b]{0.49\textwidth}
    \pdftooltip{\includegraphics[width=\textwidth]{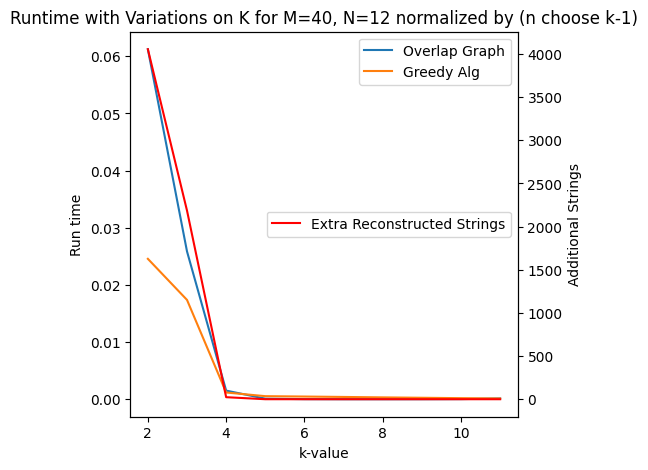}}{A plot for $n=12, m=40$ as $k$ varies, with runtimes divided by $n$-choose-$(k-1)$. The number of extra reconstructed strings starts at 4096 at $k=2$ (the point of no information), then drops sharply linear until $k=4$ where it is almost 0, and at $k=5$ it reaches 0, the point of perfect information. The overlap graph algorithm is faster when $k \geq 4$, but both algorithms' normalized runtimes follow a trend very similar to that of the extra reconstructed strings. The normalized runtime of the overlap graph is almost indistinguishable (up to scale) from the number of extra reconstructed strings.  See Appendix~\ref{app:data} for the full data table and our git repo for the data in CSV format.}
    \end{subfigure}
\begin{subfigure}[b]{0.49\textwidth}
    \pdftooltip{\includegraphics[width=\textwidth]{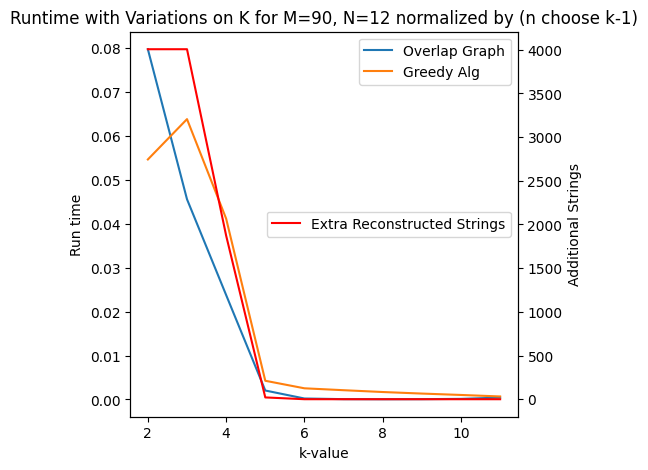}}{A plot for $n=12, m=90$ as $k$ varies, with runtimes divided by $n$-choose-$(k-1)$. The number of extra reconstructed strings starts at 4096 at $k=2$ and $3$ (the point of no information), then drops sharply linear until $k=5$ where it reaches 0, so is at the point of perfect reconstruction. The overlap graph algorithm is faster when $k \geq 3$, but both algorithms' normalized runtimes follow a trend very similar to that of the extra reconstructed strings. The greedy algorithm runtime for $k=3$ is larger than at $k=2$, but then drops off rapdily as the number of extra reconstructed strings drops. See Appendix~\ref{app:data} for the full data table and our git repo for the data in CSV format.}
    \end{subfigure}

    \begin{subfigure}[b]{0.49\textwidth}
    \pdftooltip{\includegraphics[width=\textwidth]{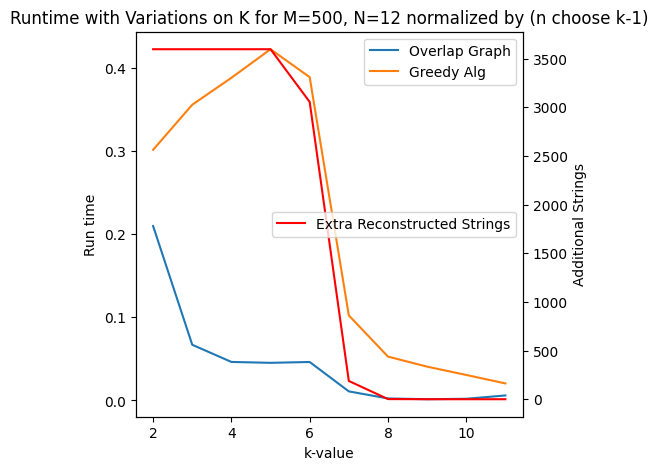}}{A plot for $n=12, m=500$ as $k$ varies, with runtimes divided by $n$-choose-$(k-1)$. The number of extra reconstructed strings starts at 4096 at $k=2$ through $5$ (the point of no information), then drops sharply until $k=8$ where it reaches 0, so is at the point of perfect reconstruction. The overlap graph algorithm is much faster than the greedy algorithm. The greedy algorithm's runtime initially increases for $k=2$ to $5$, then decreases sharply, lining up very closely with the trend in the number of extra reconstructed strings. See Appendix~\ref{app:data} for the full data table and our git repo for the data in CSV format.}
    \end{subfigure}
    \begin{subfigure}[b]{0.49\textwidth}

    \pdftooltip{\includegraphics[width=\textwidth]{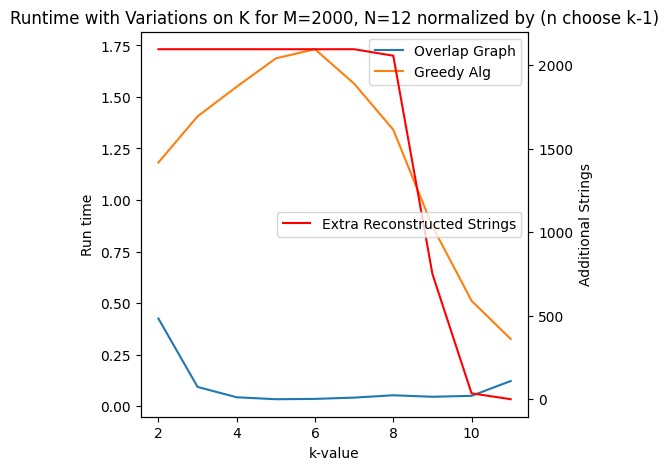}}{A plot for $n=12, m=2000$ as $k$ varies, with runtimes divided by $n$-choose-$(k-1)$. The number of extra reconstructed strings starts at 4096 at $k=2$ through $7$ (the point of no information), then drops slightly at $k=8$, sharply through $k=10$, then very slightly to the point of perfect information at $k=11$The overlap graph algorithm is much faster than the greedy algorithm. The greedy algorithm's runtime initially increases for $k=2$ to $6$, then decreases sharply, at a similar but slightly less steep trend than the number of extra reconstructed strings. See Appendix~\ref{app:data} for the full data table and our git repo for the data in CSV format.}
    \end{subfigure}
    
    \caption{Comparing the runtimes of the greedy algorithm and overlap graph algorithm, normalized by $\binom{n}{k-1}$, at $n=12$ and all values of $k$. We simultaneously plot the number of extra reconstructed strings (beyond those in $S$); one sees that to some extent both algorithms---but as $m$ gets large much moreso the greedy algorithm---are heavily governed by the latter parameter. From left to right, top to bottom: $m=40,90,500,2000$.}
    \label{fig:compareKNormalized}
\end{figure}

\section{Conclusion and future directions}
In this paper we introduced a number of computational questions in the study of higher-order interactions, proposed a new algorithm that is efficient in practice for these problems, and analyzed its runtime both analytically and experimentally. In addition to open questions posed within our work (e.g. Open Questions~\ref{open:perfect} and \ref{open:smallK}), these tools open up many further questions for exploration.

We have come across several variations of these questions in our research, not yet explored. For example, what if we allow incomplete strings in our data set? Perhaps the data is meant to represent a survey response where a participant did not answer every question, or a system only part of which could be observed. In the case of strings with missing characters, what even is a good definition of ``perfect reconstruction''? Can our algorithms be adapted to  this setting?

Larger alphabet sizes are likely to come up in applications, such as survey response data with numerous options, or modeling disease transmission with multiple states (infected, recovered, asymptomatic). The basic idea of the overlap graph and many of the results in this paper still hold for arbitrary alphabet sizes, but the complexity increases with the alphabet size. Can we do better?

Another direction we find particularly promising is leveraging sparsity in the data, since most real data sets will be quite sparse. In \S\ref{app:sparse} we show how sparsity can have an immediate impact on some of the questions we study here, and it would be interesting to take further advantage of sparsity.

\section*{Acknowledgments}
JAG thanks Yoav Kallus for many discussions over many months around 2015 (sic!) that included Definition~\ref{def:recon}. 
The authors would like to thank Jessica Flack for inspiring conversations that eventually led to this project. JAG would like to thank Gülce Kardeş for teaching him about $k$-limits in \cite{HJP,GRSS}. 
Both authors' work was funded in part by NSF CAREER award CCF-2047756, and E.~T. was additionally funded by funds from the CU Boulder Interdisciplinary Quantitative Biology program.

 \bibliographystyle{alphaurl}
 \bibliography{citations}

@article{IP01,
  author       = {Russell Impagliazzo and
                  Ramamohan Paturi},
  title        = {On the Complexity of k-{SAT}},
  journal      = {J. Comput. Syst. Sci.},
  volume       = {62},
  number       = {2},
  pages        = {367--375},
  year         = {2001},
  doi          = {10.1006/JCSS.2000.1727},
}

@misc{UedaNagao,
 author = {Ueda, N. and Nagao, T.},
title = {{NP}-completeness results for {NONOGRAM} via parsimonious reductions},
year = {1996},
howpublished = {Tech. Report TR96-0008, Dept. of Computer Science, Tokyo Institute of Technology},
}

@article{YatoSeta,
  author       = {Takayuki Yato and
                  Takahiro Seta},
  title        = {Complexity and Completeness of Finding Another Solution and Its Application
                  to Puzzles},
  journal      = {{IEICE} Trans. Fundam. Electron. Commun. Comput. Sci.},
  volume       = {86-A},
  number       = {5},
  pages        = {1052--1060},
  year         = {2003},
  url          = {http://search.ieice.org/bin/summary.php?id=e86-a\_5\_1052},
}

@Article{Feder,
 Author = {Feder, Tom{\'a}s},
 Title = {Network flow and 2-satisfiability},
 FJournal = {Algorithmica},
 Journal = {Algorithmica},
 ISSN = {0178-4617},
 Volume = {11},
 Number = {3},
 Pages = {291--319},
 Year = {1994},
 Language = {English},
 DOI = {10.1007/BF01240738},
 zbMATH = {529944},
 Zbl = {0795.68097}
}

@inproceedings{GRSS,
  author       = {Mika G{\"{o}}{\"{o}}s and
                  Artur Riazanov and
                  Anastasia Sofronova and
                  Dmitry Sokolov},
  title        = {Top-Down Lower Bounds for Depth-Four Circuits},
  booktitle    = {64th {IEEE} Annual Symposium on Foundations of Computer Science, {FOCS}
                  2023},
  pages        = {1048--1055},
  publisher    = {{IEEE}},
  year         = {2023},
  doi          = {10.1109/FOCS57990.2023.00063},
}

@Article{HJP,
 Author = {H{\aa}stad, Johan and Jukna, S. and Pudl{\'a}k, P.},
 Title = {Top-down lower bounds for depth-three circuits},
 FJournal = {Computational Complexity},
 Journal = {Comput. Complexity},
 ISSN = {1016-3328},
 Volume = {5},
 Number = {2},
 Pages = {99--112},
 Year = {1995},
 Language = {English},
 DOI = {10.1007/BF01268140},
 zbMATH = {829922},
 Zbl = {0838.68056}
}

@book{sipser,
author = {Michael Sipser},
title = {Introduction to the Theory of Computation},
year = {2012},
publisher = {Cengage Learning},
note = {3rd Edition},
}

@article{BattistonEtAl,
author = {Federico Battiston and Giulia Cencetti and Iacopo Iacopini and Vito Latora and Maxime Lucas and Alice Patania and Jean-Gabriel Young and Giovanni Petri},
title = {Networks beyond pairwise interactions: Structure and dynamics},
journal = {Physics Reports},
volume = {874},
pages = {1-92},
year = {2020},
issn = {0370-1573},
doi = {10.1016/j.physrep.2020.05.004},
}

@article{BickEtAl,
author = {Bick, Christian and Gross, Elizabeth and Harrington, Heather A. and Schaub, Michael T.},
title = {What Are Higher-Order Networks?},
journal = {SIAM Review},
volume = {65},
number = {3},
pages = {686-731},
year = {2023},
doi = {10.1137/21M1414024},
}

@article{TorresEtAl,
author = {Torres, Leo and Blevins, Ann S. and Bassett, Danielle and Eliassi-Rad, Tina},
title = {The Why, How, and When of Representations for Complex Systems},
journal = {SIAM Review},
volume = {63},
number = {3},
pages = {435-485},
year = {2021},
doi = {10.1137/20M1355896},
}

@misc{GrochowSFIReview,
       AUTHOR = {Joshua A. Grochow},
        TITLE = {Beyond pairwise: higher-order interactions in complex systems},
         YEAR = {2022},
 HOWPUBLISHED = {John Templeton Foundation Research Review
      \href{https://www.templeton.org/discoveries/complexity}{www.templeton.org/discoveries/complexity}},
         NOTE = {Produced by the Santa Fe Institute},
}

@Book{CKS,
 Author = {Creignou, Nadia and Khanna, Sanjeev and Sudan, Madhu},
 Title = {Complexity classifications of {Boolean} constraint satisfaction problems},
 FSeries = {SIAM Monographs on Discrete Mathematics and Applications},
 Series = {SIAM Monogr. Discrete Math. Appl.},
 Volume = {7},
 ISBN = {0-89871-479-6; 978-0-89871-854-6},
 Year = {2001},
 Publisher = {Philadelphia, PA: SIAM},
 Language = {English},
 DOI = {10.1137/1.9780898718546},
 zbMATH = {1614051},
 Zbl = {0981.68058}
}

@Book{Downey,
 Author = {Downey, Rodney G. and Fellows, Michael R.},
 Title = {Fundamentals of Parameterized Complexity},
 ISBN = {978-1-4471-5558-4},
 Year = {2013},
 Publisher = {Springer London},
 Language = {English},
 DOI = {10.1007/978-1-4471-5559-1},
}

@Book{zivny,
 Author = {{\v{Z}}ivn{\'y}, Stanislav},
 Title = {The complexity of valued constraint satisfaction problems},
 FSeries = {Cognitive Technologies},
 Series = {Cogn. Technol.},
 ISSN = {1611-2482},
 ISBN = {978-3-642-33973-8; 978-3-642-33974-5},
 Year = {2012},
 Publisher = {Berlin: Springer},
 Language = {English},
 DOI = {10.1007/978-3-642-33974-5},
 zbMATH = {6147065},
 Zbl = {1283.68024}
}

@Book{petke,
 Author = {Petke, Justyna},
 Title = {Bridging constraint satisfaction and {Boolean} satisfiability},
 FSeries = {Artificial Intelligence: Foundations, Theory, and Algorithms},
 Series = {Artif. Intell.: Found. Theory Algorithms},
 ISSN = {2365-3051},
 ISBN = {978-3-319-21809-0; 978-3-319-21810-6},
 Year = {2015},
 Publisher = {Cham: Springer},
 Language = {English},
 DOI = {10.1007/978-3-319-21810-6},
 zbMATH = {6540568},
 Zbl = {1410.68008}
}

@misc{LeeDanielsFlackKrakauer,
 author = {Edward D. Lee and Bryan C. Daniels and Jessica C. Flack and David C. Krakauer},
 title = {Capturing collective conflict dynamics with sparse social circuits},
 year = {2014},
 howpublished = {arXiv:\href{https://arxiv.org/abs/1406.7720}{1406.7720 [cs.SI]}},
}

@article{DanielsKrakauerFlack,
author = {Bryan C. Daniels  and David C. Krakauer  and Jessica C. Flack },
title = {Sparse code of conflict in a primate society},
journal = {PNAS},
fjournal = {Proceedings of the National Academy of Sciences USA},
volume = {109},
number = {35},
pages = {14259--14264},
year = {2012},
doi = {10.1073/pnas.1203021109},
}

@article{rizzi19,
  title = {Overlap graphs and de {Bruijn} graphs: data structures for de novo genome assembly in the big data era},
  author = {Raffaella Rizzi and Stefano Beretta and Murray Patterson and Yuri Pirola and Marco Previtali and Gianluca Della Vedova and Paola Bonizzoni},
  year = {2019},
  month = {July},
  volume = {7(4)},
  pages = {278–292},
  doi = {10.1007/s40484-019-0181-x},
  journal = {Quantitative Biology}
}

@article{dHS,
  author       = {Ronald de Haan and
                  Stefan Szeider},
  title        = {Parameterized complexity classes beyond para-{NP}},
  journal      = {J. Comput. Syst. Sci.},
  volume       = {87},
  pages        = {16--57},
  year         = {2017},
  doi          = {10.1016/J.JCSS.2017.02.002},
}

@misc{code,
 author = {Elise Tate and Joshua A. Grochow},
 title = {Code and data for ``Reconstructing sets of strings from their k-way projections: algorithms \& complexity'' (this paper)},
howpublished = {\url{https://git.cutheory.dev/jetate/String_Reassembly/src/branch/main/IWOCA_Code}. 
},
date = {9 Mar 2025},
}

@Article{paraCompendium19,
AUTHOR = {de Haan, Ronald and Szeider, Stefan},
TITLE = {A Compendium of Parameterized Problems at Higher Levels of the Polynomial Hierarchy},
JOURNAL = {Algorithms},
VOLUME = {12},
YEAR = {2019},
NUMBER = {9},
ARTICLE-NUMBER = {188},
ISSN = {1999-4893},
DOI = {10.3390/a12090188}
}

@inproceedings{Ko1995,
author="Ko, Ker-I
and Lin, Chih-Long",
editor="Du, Ding-Zhu
and Pardalos, Panos M.",
title="On the Complexity of Min-Max Optimization Problems and their Approximation",
bookTitle="Minimax and Applications",
year="1995",
publisher="Springer US",
address="Boston, MA",
pages="219--239",
isbn="978-1-4613-3557-3",
doi="10.1007/978-1-4613-3557-3_15",
}

@article{Papadimitriou84,
title = {The complexity of facets (and some facets of complexity)},
journal = {J. Comput. Syst. Sci.},
volume = {28},
number = {2},
pages = {244-259},
year = {1984},
issn = {0022-0000},
doi = {10.1016/0022-0000(84)90068-0},
author = {C.H. Papadimitriou and M. Yannakakis}
}

@article{Berman75,
author = {Berman, L. and Hartmanis, J.},
title = {On Isomorphisms and Density of {NP} and Other Complete Sets},
journal = {SIAM J. Comput.},
volume = {6},
number = {2},
pages = {305-322},
year = {1977},
doi = {10.1137/0206023},
}

@InProceedings{Calabro09,
author="Calabro, Chris
and Impagliazzo, Russell
and Paturi, Ramamohan",
editor="Chen, Jianer
and Fomin, Fedor V.",
title="The Complexity of Satisfiability of Small Depth Circuits",
booktitle="Parameterized and Exact Computation",
year="2009",
publisher="Springer Berlin Heidelberg",
address="Berlin, Heidelberg",
pages="75--85",
isbn="978-3-642-11269-0",
doi={10.1007/978-3-642-11269-0_6},
}

@inproceedings{Khot02,
author = {Khot, Subhash},
title = {On the power of unique 2-prover 1-round games},
year = {2002},
isbn = {1581134959},
publisher = {ACM},
doi = {10.1145/509907.510017},
booktitle = {Proc. 34th ACM Symp. Theory Comput. (STOC)},
pages = {767–775},
numpages = {9},
}

@article{KhotRegev,
  author       = {Subhash Khot and
                  Oded Regev},
  title        = {Vertex cover might be hard to approximate to within 2-epsilon},
  journal      = {J. Comput. Syst. Sci.},
  volume       = {74},
  number       = {3},
  pages        = {335--349},
  year         = {2008},
  doi          = {10.1016/J.JCSS.2007.06.019},
}

@article{CyganEtAl,
  author       = {Marek Cygan and
                  Holger Dell and
                  Daniel Lokshtanov and
                  D{\'{a}}niel Marx and
                  Jesper Nederlof and
                  Yoshio Okamoto and
                  Ramamohan Paturi and
                  Saket Saurabh and
                  Magnus Wahlstr{\"{o}}m},
  title        = {On Problems as Hard as {CNF-SAT}},
  journal      = {{ACM} Trans. Algorithms},
  volume       = {12},
  number       = {3},
  pages        = {41:1--41:24},
  year         = {2016},
  doi          = {10.1145/2925416},
}

\newpage

\appendix

\section{Details of the greedy algorithm for comparison} \label{app:greedy}
The greedy algorithm we compared against works as follows.

The greedy algorithm builds a list of partial reconstructed strings, that are consistent with all $\binom{i}{k}$ $k$-windows in the first $i$ indices $\{0,\dotsc,i-1\}$, for $i=k,\dotsc,n$, by induction on $i$. When $i=k$, we start by taking the substrings from the window with indices $0 \dots k-1$. Once we have the list of partial strings up to $i$, we then attempt to add 0 onto the end of each string, checking first if there are any disagreements among the windows with some subset of $k-1$ indices from the list $0 \dots i-1$ as well as the new index position, $i$. If there are no disagreements, we add the new string that had a 0 appended. We then do the same by attempting to append 1 instead of 0. We continue adding one index at a time in this style, looking only at windows that contain our new index $i$ and one of the $i - 1 \choose k - 1$ many $(k-1)$-subsets of $\{0,\dotsc,i-1\}$ at each step. 

In pseudocode (pseudo-C++17), the greedy algorithm we use is:

\begin{verbatim}
greedyApproach(vector<string> data, int n, int m, int k)
{
  vector<string> reconstructedData;

  // store the values for the first window of size k
  for (int i = 0; i < m; i++) {
    string sub = data.at(i).substr(0, k);
    if (sub not in reconstructedData) {
      reconstructedData.push_back(sub);    
    }
  }

  // for indices k to n-1... (each index to add)
  for (int i = k; i <= n - 1; i++) {
    vector<string> workingSet;
    // for each of the current strings in the reconstructed data, consider adding the next index
    for (int j = 0; j < reconstructedData.size(); j++) {
      // check against all windows with i and some subset of previous indices
      extendZero = reconstructedData.at(j) + "0";
      extendOne = reconstructedData.at(j) + "1";
      if (greedyCheck(data, i, k, extendZero)) {
        workingSet.push_back(extendZero);
      }
      if (greedyCheck(data, i, k, extendOne)) {
        workingSet.push_back(extendOne);
      }
      workingSet.eraseAt(j); // remove old string      
    }
    reconstructedData = workingSet;
  }
}
\end{verbatim}

\section{Additional details on sparse data sets}\label{app:sparse}
One direction we think is particularly interesting based on real-world data sets is the sparse setting: where, say, there are \emph{many} more 0s than 1s in the strings in our data set (or vice versa). When our data set is particularly sparse, the values for the Point of No Reconstruction and the Point of Perfect Reconstruction are heavily affected by the ratio of 0's to 1's. 

Consider the example of the standard basis vectors $e_1, \dotsc, e_n$, $e_i = (0,0,\dotsc,1,\dotsc,0)$: the all-0 string is reconstructed because every $(n-1)$-way projection contained $0^{n-1}$ as a substring.
This same phenomenon will occur in any data set where $0^k$ is present in all $k$-way projections. If every column has at least a $(k-1)/k$ fraction of 0s, then by a pigeonhole argument, every $k$-way projection must contain $0^k$, so $0^n$ is in $\Recon_k(S)$.

This logic can be further generalized to a blend of sparse and dense data columns; consider a data set where all columns have majority 0's except one, which has majority 1's. This data set would have to produce the all zero string with a 1 in the index value associated with the column that had majority 1's. 

This can ultimately be generalized as: take the count of minority digits, $d$ to total digits in a column as $\frac{d}{m}$ and sort these ratios in decreasing order such that $\frac{d_1}{m} \geq \frac{d_2}{m} \geq \dots \geq \frac{d_n}{m}$. We can then sum the fractions until we reach the maximum $i$ for which $\sum_{i} \frac{d_i}{m} < 1$; this value of $i$ is a bounds on $k$ for the Point of Perfect Reconstruction, and the string that would be incorrectly reconstructed is simply the most common digit for each index. 

Given the constraints imposed by having very unbalanced columns, can we take advantage of that in our algorithms?

\section{Data for experimental outcomes} \label{app:data}
\small

This data can be found in comma-separated value plaintext at \cite{code}.

\begin{longtable}{|c|c|c|c|c|c|c|}
\caption{Runtime data of overlap and greedy algorithms.}
\label{tab:MK}\\
\hline
$n$  & $m$  & $k$  & \parbox{1in}{Overlap graph \\ runtime avg(s)} & \parbox{1in}{Greedy \\ runtime avg(s)} & \parbox{1in}{Overlap graph \\ runtime stdev(s)} & \parbox{1in}{Greedy \\ runtime stdev(s)} \\ \hline
10 & 30 & 2 & 0.118809 & 0.044797 & 0.006029 & 0.006029 \\
10 & 30 & 3 & 0.155787 & 0.103503 & 0.040238 & 0.040238 \\
10 & 30 & 4 & 0.033638 & 0.042879 & 0.008681 & 0.008681 \\
10 & 30 & 5 & 0.011469 & 0.049065 & 0.003399 & 0.003399 \\
10 & 30 & 6 & 0.005732 & 0.049671 & 0.001396 & 0.001396 \\
10 & 30 & 7 & 0.004646 & 0.033937 & 0.001058 & 0.001058 \\
10 & 30 & 8 & 0.004949 & 0.015121 & 0.000747 & 0.000747 \\
10 & 30 & 9 & 0.005935 & 0.005580 & 0.000214 & 0.000214 \\
12 & 10 & 2 & 0.120791 & 0.014946 & 0.038448 & 0.038448 \\
12 & 10 & 3 & 0.013666 & 0.004076 & 0.005123 & 0.005123 \\
12 & 10 & 4 & 0.003749 & 0.008665 & 0.001253 & 0.001253 \\
12 & 10 & 5 & 0.001958 & 0.015623 & 0.000263 & 0.000263 \\
12 & 10 & 6 & 0.001643 & 0.022287 & 0.000341 & 0.000341 \\
12 & 10 & 7 & 0.001405 & 0.023092 & 0.000259 & 0.000259 \\
12 & 10 & 8 & 0.001291 & 0.015935 & 0.000133 & 0.000133 \\
12 & 10 & 9 & 0.001279 & 0.007930 & 0.000091 & 0.000091 \\
12 & 10 & 10 & 0.002101 & 0.003911 & 0.000188 & 0.000188 \\
12 & 10 & 11 & 0.002251 & 0.001008 & 0.000102 & 0.000102 \\
12 & 20 & 2 & 0.541712 & 0.131502 & 0.088984 & 0.088984 \\
12 & 20 & 3 & 0.182222 & 0.053629 & 0.054221 & 0.054221 \\
12 & 20 & 4 & 0.023491 & 0.035654 & 0.005844 & 0.005844 \\
12 & 20 & 5 & 0.007582 & 0.066045 & 0.001260 & 0.001260 \\
12 & 20 & 6 & 0.004655 & 0.095069 & 0.001080 & 0.001080 \\
12 & 20 & 7 & 0.003479 & 0.094683 & 0.000655 & 0.000655 \\
12 & 20 & 8 & 0.002998 & 0.066049 & 0.000489 & 0.000489 \\
12 & 20 & 9 & 0.002867 & 0.032677 & 0.000622 & 0.000622 \\
12 & 20 & 10 & 0.002909 & 0.011069 & 0.000473 & 0.000473 \\
12 & 20 & 11 & 0.004368 & 0.003241 & 0.000531 & 0.000531 \\
12 & 30 & 2 & 0.675664 & 0.218843 & 0.036640 & 0.036640 \\
12 & 30 & 3 & 0.860821 & 0.452921 & 0.272321 & 0.272321 \\
12 & 30 & 4 & 0.113619 & 0.094576 & 0.022463 & 0.022463 \\
12 & 30 & 5 & 0.022348 & 0.150567 & 0.004590 & 0.004590 \\
12 & 30 & 6 & 0.010643 & 0.214816 & 0.002714 & 0.002714 \\
12 & 30 & 7 & 0.006608 & 0.214974 & 0.001451 & 0.001451 \\
12 & 30 & 8 & 0.005365 & 0.149519 & 0.001009 & 0.001009 \\
12 & 30 & 9 & 0.004868 & 0.073995 & 0.001167 & 0.001167 \\
12 & 30 & 10 & 0.005754 & 0.026320 & 0.001156 & 0.001156 \\
12 & 30 & 11 & 0.007558 & 0.006692 & 0.000621 & 0.000621 \\
12 & 40 & 2 & 0.734910 & 0.294899 & 0.027142 & 0.027142 \\
12 & 40 & 3 & 1.707456 & 1.147517 & 0.407452 & 0.407452 \\
12 & 40 & 4 & 0.333967 & 0.261492 & 0.072070 & 0.072070 \\
12 & 40 & 5 & 0.056233 & 0.271161 & 0.010262 & 0.010262 \\
12 & 40 & 6 & 0.022767 & 0.381334 & 0.005324 & 0.005324 \\
12 & 40 & 7 & 0.012934 & 0.379374 & 0.003703 & 0.003703 \\
12 & 40 & 8 & 0.009913 & 0.264329 & 0.002679 & 0.002679 \\
12 & 40 & 9 & 0.008263 & 0.131825 & 0.001861 & 0.001861 \\
12 & 40 & 10 & 0.007434 & 0.043607 & 0.001210 & 0.001210 \\
12 & 40 & 11 & 0.009919 & 0.009148 & 0.000551 & 0.000551 \\
12 & 50 & 2 & 0.779734 & 0.367064 & 0.024202 & 0.024202 \\
12 & 50 & 3 & 2.421463 & 2.009077 & 0.354848 & 0.354848 \\
12 & 50 & 4 & 0.774133 & 0.644271 & 0.206669 & 0.206669 \\
12 & 50 & 5 & 0.119219 & 0.434200 & 0.021888 & 0.021888 \\
12 & 50 & 6 & 0.037267 & 0.600390 & 0.005357 & 0.005357 \\
12 & 50 & 7 & 0.021841 & 0.593723 & 0.004912 & 0.004912 \\
12 & 50 & 8 & 0.014185 & 0.411490 & 0.002751 & 0.002751 \\
12 & 50 & 9 & 0.011684 & 0.204612 & 0.002601 & 0.002601 \\
12 & 50 & 10 & 0.011258 & 0.069297 & 0.001919 & 0.001919 \\
12 & 50 & 11 & 0.013689 & 0.013973 & 0.000568 & 0.000568 \\
12 & 60 & 2 & 0.816233 & 0.439537 & 0.018423 & 0.018423 \\
12 & 60 & 3 & 2.693961 & 2.679080 & 0.205129 & 0.205129 \\ 
12 & 60 & 4 & 1.592846 & 1.760022 & 0.388749 & 0.388749 \\
12 & 60 & 5 & 0.215926 & 0.662540 & 0.027660 & 0.027660 \\
12 & 60 & 6 & 0.051899 & 0.869770 & 0.006892 & 0.006892 \\
12 & 60 & 7 & 0.029465 & 0.856174 & 0.003905 & 0.003905 \\
12 & 60 & 8 & 0.020535 & 0.592009 & 0.003707 & 0.003707 \\
12 & 60 & 9 & 0.016413 & 0.294055 & 0.003000 & 0.003000 \\
12 & 60 & 10 & 0.015973 & 0.100025 & 0.001935 & 0.001935 \\
12 & 60 & 11 & 0.016608 & 0.020095 & 0.001954 & 0.001954 \\
12 & 70 & 2 & 0.878817 & 0.512407 & 0.027310 & 0.027310 \\
12 & 70 & 3 & 2.907783 & 3.247799 & 0.118709 & 0.118709 \\
12 & 70 & 4 & 2.544709 & 3.382021 & 0.602658 & 0.602658 \\
12 & 70 & 5 & 0.375773 & 0.970210 & 0.056235 & 0.056235 \\
12 & 70 & 6 & 0.088691 & 1.193889 & 0.008203 & 0.008203 \\
12 & 70 & 7 & 0.039939 & 1.167277 & 0.004142 & 0.004142 \\
12 & 70 & 8 & 0.027236 & 0.805773 & 0.003334 & 0.003334 \\
12 & 70 & 9 & 0.022542 & 0.398278 & 0.003602 & 0.003602 \\
12 & 70 & 10 & 0.021868 & 0.136145 & 0.002173 & 0.002173 \\
12 & 70 & 11 & 0.022725 & 0.027426 & 0.001507 & 0.001507 \\
12 & 80 & 2 & 0.916063 & 0.584089 & 0.020526 & 0.020526 \\
12 & 80 & 3 & 3.000175 & 3.744063 & 0.116326 & 0.116326 \\
12 & 80 & 4 & 4.039950 & 5.902573 & 0.739820 & 0.739820 \\
12 & 80 & 5 & 0.605805 & 1.405045 & 0.095085 & 0.095085 \\
12 & 80 & 6 & 0.122675 & 1.567368 & 0.011166 & 0.011166 \\
12 & 80 & 7 & 0.048442 & 1.526201 & 0.004978 & 0.004978 \\
12 & 80 & 8 & 0.031619 & 1.052121 & 0.003073 & 0.003073 \\
12 & 80 & 9 & 0.025010 & 0.520162 & 0.002910 & 0.002910 \\
12 & 80 & 10 & 0.023835 & 0.176667 & 0.002033 & 0.002033 \\
12 & 80 & 11 & 0.022104 & 0.035359 & 0.003123 & 0.003123 \\
12 & 90 & 2 & 0.956820 & 0.655612 & 0.026650 & 0.026650 \\
12 & 90 & 3 & 3.007996 & 4.211977 & 0.105778 & 0.105778 \\
12 & 90 & 4 & 5.229099 & 9.055124 & 1.058499 & 1.058499 \\
12 & 90 & 5 & 1.004792 & 2.103427 & 0.150663 & 0.150663 \\
12 & 90 & 6 & 0.167558 & 2.006025 & 0.019330 & 0.019330 \\
12 & 90 & 7 & 0.060676 & 1.941134 & 0.005015 & 0.005015 \\
12 & 90 & 8 & 0.037510 & 1.336043 & 0.004099 & 0.004099 \\
12 & 90 & 9 & 0.029891 & 0.660568 & 0.003790 & 0.003790 \\
12 & 90 & 10 & 0.028128 & 0.223883 & 0.002181 & 0.002181 \\
12 & 90 & 11 & 0.027881 & 0.044626 & 0.002888 & 0.002888 \\
12 & 100 & 2 & 1.002506 & 0.729048 & 0.028012 & 0.028012 \\
12 & 100 & 3 & 3.052822 & 4.684431 & 0.112754 & 0.112754 \\
12 & 100 & 4 & 6.997932 & 12.947197 & 0.998384 & 0.998384 \\
12 & 100 & 5 & 1.396664 & 2.919842 & 0.231179 & 0.231179 \\
12 & 100 & 6 & 0.228576 & 2.493792 & 0.023585 & 0.023585 \\
12 & 100 & 7 & 0.081067 & 2.396268 & 0.007638 & 0.007638 \\
12 & 100 & 8 & 0.043425 & 1.643364 & 0.003116 & 0.003116 \\
12 & 100 & 9 & 0.034354 & 0.815149 & 0.003160 & 0.003160 \\
12 & 100 & 10 & 0.032382 & 0.275801 & 0.001556 & 0.001556 \\
12 & 100 & 11 & 0.029412 & 0.055729 & 0.003257 & 0.003257 \\
12 & 500 & 2 & 2.519556 & 3.620742 & 0.028139 & 0.028139 \\
12 & 500 & 3 & 4.420212 & 23.498960 & 0.118381 & 0.118381 \\
12 & 500 & 4 & 10.189787 & 85.455540 & 0.262634 & 0.262634 \\
12 & 500 & 5 & 22.394610 & 209.313500 & 0.472883 & 0.472883 \\
12 & 500 & 6 & 36.633370 & 308.249600 & 2.110679 & 2.110679 \\
12 & 500 & 7 & 10.027627 & 94.567290 & 0.639925 & 0.639925 \\
12 & 500 & 8 & 1.864623 & 41.782080 & 0.121197 & 0.121197 \\
12 & 500 & 9 & 0.702858 & 20.111190 & 0.031041 & 0.031041 \\
12 & 500 & 10 & 0.428198 & 6.741600 & 0.008602 & 0.008602 \\
12 & 500 & 11 & 0.393686 & 1.348636 & 0.004891 & 0.004891 \\
12 & 1000 & 2 & 3.968106 & 7.223571 & 0.025157 & 0.025157 \\
12 & 1000 & 3 & 5.526174 & 46.966440 & 0.087395 & 0.087395 \\
12 & 1000 & 4 & 10.431460 & 171.134800 & 0.172144 & 0.172144 \\
12 & 1000 & 5 & 20.964530 & 421.369700 & 0.321329 & 0.321329 \\
12 & 1000 & 6 & 37.273470 & 693.403200 & 0.676447 & 0.676447 \\
12 & 1000 & 7 & 49.368140 & 706.226600 & 1.170888 & 1.170888 \\
12 & 1000 & 8 & 17.127330 & 244.114100 & 1.229378 & 1.229378 \\
12 & 1000 & 9 & 4.107391 & 83.590200 & 0.114788 & 0.114788 \\
12 & 1000 & 10 & 2.006797 & 27.200940 & 0.021436 & 0.021436 \\
12 & 1000 & 11 & 1.470922 & 5.431625 & 0.010899 & 0.010899 \\
12 & 1500 & 2 & 4.893284 & 10.827800 & 0.017423 & 0.017423 \\
12 & 1500 & 3 & 6.195878 & 70.662190 & 0.071729 & 0.071729 \\
12 & 1500 & 4 & 10.260198 & 257.414000 & 0.177076 & 0.177076 \\
12 & 1500 & 5 & 19.197140 & 633.709500 & 0.378227 & 0.378227 \\
12 & 1500 & 6 & 32.881030 & 1023.964000 & 0.657284 & 0.657284 \\
12 & 1500 & 7 & 46.220780 & 1093.901000 & 0.590068 & 0.590068 \\
12 & 1500 & 8 & 41.510060 & 721.021000 & 1.516087 & 1.516087 \\
12 & 1500 & 9 & 11.901450 & 211.916900 & 0.296781 & 0.296781 \\
12 & 1500 & 10 & 5.676148 & 61.803300 & 0.119352 & 0.119352 \\
12 & 1500 & 11 & 3.734903 & 12.157400 & 0.041770 & 0.041770 \\
12 & 2000 & 2 & 5.099039 & 14.185750 & 0.017339 & 0.017339 \\
12 & 2000 & 3 & 6.162363 & 92.788470 & 0.058511 & 0.058511 \\
12 & 2000 & 4 & 9.476573 & 340.929500 & 0.110094 & 0.110094 \\
12 & 2000 & 5 & 16.667190 & 835.140600 & 0.216039 & 0.216039 \\
12 & 2000 & 6 & 27.880520 & 1371.374000 & 0.624081 & 0.624081 \\
12 & 2000 & 7 & 38.323530 & 1445.620000 & 0.623282 & 0.623282 \\
12 & 2000 & 8 & 41.836570 & 1062.018000 & 0.392293 & 0.392293 \\
12 & 2000 & 9 & 22.558390 & 431.830400 & 0.774261 & 0.774261 \\
12 & 2000 & 10 & 11.024740 & 112.454000 & 0.122618 & 0.122618 \\
12 & 2000 & 11 & 8.052963 & 21.508270 & 0.103578 & 0.103578 \\
14 & 30 & 2 & 3.731794 & 1.036507 & 0.287231 & 0.287231 \\
14 & 30 & 3 & 3.086077 & 1.142078 & 1.068723 & 1.068723 \\
14 & 30 & 4 & 0.350067 & 0.175320 & 0.079797 & 0.079797 \\
14 & 30 & 5 & 0.046482 & 0.368453 & 0.009590 & 0.009590 \\
14 & 30 & 6 & 0.020922 & 0.683616 & 0.004622 & 0.004622 \\
14 & 30 & 7 & 0.012017 & 0.904974 & 0.002970 & 0.002970 \\
14 & 30 & 8 & 0.009316 & 0.882247 & 0.002176 & 0.002176 \\
14 & 30 & 9 & 0.007557 & 0.647876 & 0.001681 & 0.001681 \\
14 & 30 & 10 & 0.007273 & 0.358885 & 0.001784 & 0.001784 \\
14 & 30 & 11 & 0.006475 & 0.146859 & 0.001730 & 0.001730 \\
14 & 30 & 12 & 0.006113 & 0.039918 & 0.001377 & 0.001377 \\
14 & 30 & 13 & 0.008741 & 0.007911 & 0.001666 & 0.001666 \\
16 & 30 & 2 & 19.814633 & 4.837633 & 2.162247 & 2.162247 \\
16 & 30 & 3 & 12.394322 & 3.398693 & 3.869013 & 3.869013 \\
16 & 30 & 4 & 1.175649 & 0.298819 & 0.351317 & 0.351317 \\
16 & 30 & 5 & 0.104782 & 0.791696 & 0.020072 & 0.020072 \\
16 & 30 & 6 & 0.029752 & 1.805207 & 0.006063 & 0.006063 \\
16 & 30 & 7 & 0.014408 & 2.996054 & 0.003782 & 0.003782 \\
16 & 30 & 8 & 0.011086 & 3.733377 & 0.002768 & 0.002768 \\
16 & 30 & 9 & 0.009822 & 3.642316 & 0.002366 & 0.002366 \\
16 & 30 & 10 & 0.008204 & 2.805020 & 0.001850 & 0.001850 \\
16 & 30 & 11 & 0.008303 & 1.689993 & 0.001924 & 0.001924 \\
16 & 30 & 12 & 0.008223 & 0.773571 & 0.002090 & 0.002090 \\
16 & 30 & 13 & 0.008053 & 0.262174 & 0.002132 & 0.002132 \\
16 & 30 & 14 & 0.007722 & 0.062382 & 0.001772 & 0.001772 \\
16 & 30 & 15 & 0.010204 & 0.010163 & 0.001822 & 0.001822 \\
18 & 30 & 2 & 109.934667 & 22.857617 & 3.804745 & 3.804745 \\
18 & 30 & 3 & 50.454777 & 8.125150 & 25.096455 & 25.096455 \\
18 & 30 & 4 & 3.387042 & 0.480716 & 1.016212 & 1.016212 \\
18 & 30 & 5 & 0.233913 & 1.478137 & 0.051633 & 0.051633 \\
18 & 30 & 6 & 0.042572 & 4.166179 & 0.008330 & 0.008330 \\
18 & 30 & 7 & 0.021793 & 8.286597 & 0.006813 & 0.006813 \\
18 & 30 & 8 & 0.014039 & 12.646790 & 0.004095 & 0.004095 \\
18 & 30 & 9 & 0.014109 & 15.401710 & 0.003638 & 0.003638 \\
18 & 30 & 10 & 0.010921 & 15.265493 & 0.003028 & 0.003028 \\
18 & 30 & 11 & 0.011171 & 11.386507 & 0.002167 & 0.002167 \\
18 & 30 & 12 & 0.009320 & 7.225315 & 0.001483 & 0.001483 \\
18 & 30 & 13 & 0.009457 & 3.621550 & 0.001812 & 0.001812 \\
18 & 30 & 14 & 0.011758 & 1.400735 & 0.002345 & 0.002345 \\
18 & 30 & 15 & 0.011163 & 0.407041 & 0.002315 & 0.002315 \\
18 & 30 & 16 & 0.010699 & 0.085244 & 0.002586 & 0.002586 \\
18 & 30 & 17 & 0.014418 & 0.012287 & 0.002549 & 0.002549 \\
20 & 30 & 2 & 556.622875 & 93.828213 & 41.056464 & 41.056464 \\
20 & 30 & 5 & 0.491720 & 2.662911 & 0.146872 & 0.146872 \\
 \hline
\end{longtable}

\end{document}